\documentclass[journal]{IEEEtran}
\usepackage[section]{placeins}
\usepackage{cite}
\usepackage{array}
\usepackage{textcomp}
\usepackage{stfloats}
\usepackage{url}
\usepackage{verbatim}
\usepackage{graphicx}
\usepackage{booktabs}
\usepackage{subfigure}
\usepackage{xcolor}
\usepackage{amsthm, amsfonts}
\usepackage{bbm}
\usepackage{mathrsfs}
\usepackage{multicol}
\usepackage{epstopdf}
\newtheorem{theorem}{Theorem}
\newtheorem{lemma}[theorem]{Lemma}

\usepackage{balance}
\usepackage{tikz}
\usepackage{mathtools}
\allowdisplaybreaks
\usepackage{microtype}
\usepackage{amsmath,amssymb,mathtools}
\newcommand{\mat}[1]{\mathbf{#1}}     
\newcommand{\vect}[1]{\mathbf{#1}}

\usetikzlibrary{positioning,calc,arrows.meta,shapes,fit,decorations.pathreplacing}

\usepackage{algorithm}
\usepackage[noend]{algpseudocode}
\algrenewcommand\algorithmicrequire{\textbf{Inputs:}}
\algrenewcommand\algorithmicensure{\textbf{Output:}}
\algrenewcommand\algorithmiccomment[1]{\hfill{\footnotesize$\triangleright$~#1}}

\def\BibTeX{{\rm B\kern-.05em{\sc i\kern-.025em b}\kern-.08em
    T\kern-.1667em\lower.7ex\hbox{E}\kern-.125emX}}
\renewcommand{\algorithmicrequire}{\textbf{INPUT:}} 
\renewcommand{\algorithmicensure}{\textbf{OUTPUT:}}

\newcommand{\atanTwo}{\operatorname{atan2}}

\NewDocumentCommand{\HuberFunc}{m m}{
  \ensuremath{\operatorname{Huber}_{#2}\!\left(#1\right)}
}

\begin{document}
\title{Multimodal Large Language Model Enabled Robust Beamforming for HAP Downlink Communications}
\author{
Xiaoyu Xing,~\IEEEmembership{Graduate Student Member,~IEEE}, 
Peng Yang,~\IEEEmembership{Member,~IEEE}, 
Guoquan Tao, Dingyi Lu, \\ 
Zehui Xiong,~\IEEEmembership{Senior Member,~IEEE}, 
and Xianbin Cao,~\IEEEmembership{Senior Member,~IEEE}

 \thanks{
X. Xing, P. Yang, D. Lu, and X. Cao are with the School of Electronic Information Engineering, Beihang University, Beijing 100191, China. 
G. Tao is with the School of Institute of Unmanned Systems, Beihang University, Beijing 100191, China. 
Z. Xiong is with the School of Electronics, Electrical Engineering and Computer Science, Queen’s University Belfast, Belfast BT7 1NN, U.K.
}
}


\maketitle

\begin{abstract}
Small changes in high altitude platform (HAP) attitude can cause significant deviations in HAP downlink beam directions, thereby severely degrading HAP downlink communication performance. 
In this paper, we develop a multimodal large language model (LLM) enabled beamforming framework to achieve robust HAP downlink communications.
Specifically, we design a vision-language LLM (VL-LLM) that learns from multivariate flight telemetry to {forecast} short-term HAP attitudes {under platform shaking and support delay-aware proactive beam steering}.
We design an {offline forecast-error calibration procedure} to obtain upper bounds on {forecast} errors and improve the reliability of proactive {analog beam steering}.
{Based on the attitude forecasts, we proactively update the analog beamformer and propose a QoS-driven beamforming and admission method with a lightweight feasibility-enforcement step to satisfy instantaneous transmit-power and QoS requirements.}
Simulation results indicate that the designed VL-LLM can accurately capture changes in the HAP attitude and the proposed beamforming method achieves a 22.1\% higher user service ratio and a 12.5\% higher sum-rate than representative baselines.
{The measured mean and p99 total latencies are 36.24 ms and 40.13 ms, respectively, supporting practical delay-aware deployment.}

\end{abstract}


\section{Introduction}
\IEEEPARstart{H}{igh-altitude} platforms (HAPs) operate quasi-stationary in the lower stratosphere and are a promising enabler for wide-area, on-demand coverage in beyond-5G/6G networks \cite{10918743,10906539,10981845,10926897,11087614,AN2025104036}.
Compared with low-altitude unmanned aerial vehicles (UAVs), HAPs offer a much larger footprint and multi-month endurance, while their 20-25~km altitude keeps propagation delay within a few milliseconds and avoids the high deployment cost of satellites {\cite{11169418,10675381}}.
To support high data rates over long-distance downlinks, mmWave links with antenna arrays are highly attractive for HAP communications, where hybrid analog/digital beamforming can be explored to reduce radio frequency (RF) chain cost and power consumption \cite{10680080,6717211,7389996}.

A key challenge in HAP beamforming is maintaining reliable beam alignment under continuous platform attitude shaking.
However, HAPs are continuously perturbed by wind and turbulence. 
{Even small attitude shaking can tilt the array plane and misalign the beam with the line-of-sight (LoS) direction. 
This causes noticeable main-lobe gain loss for uniform planar arrays (UPAs) \cite{10927643,10736523}.}
Therefore, the steering angles must be adjusted in both azimuth and elevation using instantaneous HAP attitude information. 
Otherwise, narrow UPA beams can incur substantial main-lobe gain loss.
A practical {solution} is to {forecast short-term} HAP attitudes and {proactively} steer the beam before misalignment degrades link performance.
{However, in practical HAP downlinks, beam updates are not instantaneous. 
{Attitude forecasting}, communication optimization, and actuation together introduce non-negligible delay. 
As a result, a practical forecasting framework must support delay-aware proactive beam steering rather than idealized zero-latency compensation.
Moreover, forecast errors translate into pointing uncertainty. 
Thus, the challenge extends beyond forecasting future attitudes. 
It also requires translating imperfect forecasts into robust hybrid beamforming decisions under quality-of-service (QoS) and power constraints.
In this paper, we investigate delay-aware robust hybrid beamforming for HAP downlink communications by leveraging short-term HAP attitude forecasting.}

\subsection{State of the Art}
Overall, HAP communications face a fundamental challenge in maintaining beam alignment {across wide-area coverage}, particularly at mmWave and higher frequencies where narrow beams are {essential} \cite{10422712,ghanbari2025future}.
Even small {platform attitude shaking} can reduce array gain and degrade reliability and throughput, {thereby} motivating robust beam management.
Recent studies can be broadly grouped into short-term beam tracking, large language model (LLM)-enabled communication intelligence, and learning-assisted beamforming.

\textbf{Short-term beam tracking and {forecasting} for air-based platforms:}
A large body of work focuses on short-term beam tracking. 
In \cite{10535476}, a multi-armed bandit beam-tracking scheme for HAPs sequentially learns the best beam and achieves outage probability close to an ideal direction-of-arrival tracker while keeping complexity linear in the codebook size.
In \cite{11072409}, a lightweight global positioning system (GPS)-aided deep learning model extracts motion-pattern features via a convolutional neural network (CNN) and uses an encoder-decoder gated recurrent unit (GRU) to forecast future beam indices over multiple slots.
In \cite{10717332}, an attitude compensation-based beam-tracking (ACBT) method iteratively performs attitude compensation and beam tracking to cope with wind-induced UAV attitude {shaking}.

However, these works \cite{10535476,11072409,10717332} mainly focus on beam tracking {or forecasting accuracy, leaving the translation of forecast errors into analog beamforming uncertainty and multiuser QoS-aware beamforming decisions largely unaddressed.}
This motivates short-term proactive beam steering with robustness considerations over the next few slots.

\textbf{Foundation models and LLMs for communications:}
Recent works have begun to explore LLMs for communication-related tasks.
In \cite{10582827}, the authors discuss key challenges in enabling existing pre-trained LLMs to interpret and reason over communication system data.
In \cite{nikbakht2024tspecllmopensourcedatasetllm}, the authors construct a technical-document dataset to facilitate domain adaptation and knowledge injection for communications-oriented LLM applications.
In \cite{piovesan2024telecomlanguagemodelslarge}, the authors show that relatively small language models, when augmented with retrieval-augmented generation (RAG), can already capture and utilize communication knowledge effectively.
In \cite{bornea2024telcoragnavigatingchallengesretrievalaugmented}, the authors further refine the RAG method for specialized technical corpora to improve retrieval reliability and downstream reasoning performance for telecommunications.
In \cite{10679152}, the authors introduced a generative AI agent framework for satellite networks, in which LLMs and retrieval-augmented generation (RAG) were used for customized system modeling and a mixture-of-experts (MoE) method was employed for transmission optimization. In \cite{10681129}, the authors studied the use of large models for aerial edges through an integrated air-ground edge-cloud framework, with an emphasis on bandwidth-aware task allocation, data transmission, and model evolution for UAV-based multimodal inference. In \cite{11352781}, the authors combined channel extrapolation and generative AI for CSI generation in dynamic wireless environments, and showed that diffusion-based generative models improved robustness in time-varying communication tasks.

{However, these works \cite{10582827,nikbakht2024tspecllmopensourcedatasetllm,piovesan2024telecomlanguagemodelslarge,bornea2024telcoragnavigatingchallengesretrievalaugmented,10679152,10681129,11352781} mainly target communication knowledge utilization, customized network modeling, edge-cloud large-model collaboration, or security-oriented CSI generation.
They do not study multivariate HAP attitude forecasting for proactive beam steering or its coupling with robust hybrid beamforming.}

\textbf{Learning-assisted beamforming:}
Beamforming problems are commonly addressed using two classes of algorithmic tools.
The first class relies on model-based iterative solvers, such as weighted minimum mean-square error (WMMSE) and successive convex approximation (SCA). 
Recent WMMSE-based designs still involve multi-block alternating optimization with repeated variable updates \cite{10977774,11130524,11010920}, while SCA-based designs require iteratively forming and solving convex surrogate subproblems \cite{ZHANG2026156209,ZHU2026103001}.
As a result, runtime is often dominated by iteration counts and repeated numerical solves, which is challenging for online deployment, 
especially as the number of users and antennas grows.
The second class uses learning-based forecasters for fast inference. 
However, without explicit feasibility guarantees, such forecasters may output infeasible admission decisions or beamformers, leading to QoS violations, which can compromise reliability in online deployment \cite{10333602,11148833}.

To tackle the above issue, recent works have blended learning with optimization to accelerate resource allocation while preserving the feasibility of solutions. 
In \cite{10333602}, the authors introduced a knowledge distillation framework where an optimization routine guided a neural forecaster.
In \cite{11148833}, the authors demonstrated karush-kuhn-tucker (KKT) guided dual learning, where learning a small number of dual variables enabled structured beamformer reconstruction with lower computational complexity. 
In \cite{wang2025differentiableprojectionbasedlearnoptimize}, differentiable projection modules were used to project network outputs onto constraint sets. 
To avoid infeasible solutions, \cite{deng2025hophomeomorphicpolarlearning} developed homeomorphic polar learning using an invertible mapping that strictly satisfied convex or non-convex constraints. 

However, most learning-assisted beamforming methods \cite{10333602,wang2025differentiableprojectionbasedlearnoptimize,deng2025hophomeomorphicpolarlearning,11148833} assume that the analog beam setting is fixed or perfectly known.
They do not explicitly model the beam pointing errors caused by imperfect {forecasts}. This is critical in {forecast-aided} hybrid beamforming for HAP downlinks.

{In summary, prior works typically address beam tracking, communication intelligence, dynamic CSI generation, or learning-based beamforming in isolation. 
The joint design of multivariate HAP attitude forecasting and robust hybrid beamforming remains largely unexplored. 
This includes the explicit modeling of forecast-error-induced beam pointing uncertainty and its impact on multiuser QoS guarantees.}

\subsection{Motivations and Contributions}

Overall, most existing schemes remain reactive, and beam correction is activated only after misalignment occurs, introducing correction delay and transient array-gain loss.
This motivates a forecast-then-beamform {method}.
We {forecast} short-term HAP attitudes and {proactively perform beamforming, thereby reducing beam misalignment before it causes noticeable link-performance degradation.}

Accurate attitude {forecasting} from multivariate flight telemetry is nontrivial and is naturally a time-series forecasting problem.
In the time-series domain, recent works \cite{DBLP:journals/corr/abs-2310-09751,jin2024timellm,10892257,10.1145/3589334.3645434} show that LLMs can be repurposed for {forecasting} when inputs are appropriately encoded. 
This design reuses the strong sequence modeling ability of LLMs. 
It also allows domain instructions or priors to be injected through prompts with limited extra parameters.
Researchers \cite{wang2025timemixer,wu2023timesnet} have converted time series into image-like representations for vision-based modeling. 
This transformation exposes temporal structures as spatial patterns. 
It allows CNN/vision transformer (ViT) \cite{dosovitskiy2021an} to extract hierarchical features that capture latent temporal relations.

However, LLM-based methods \cite{jin2024timellm,10892257,10.1145/3589334.3645434} face a modality gap between continuous time-series and discrete tokens, and they are not explicitly pre-trained for fine-grained temporal patterns, which can limit forecasting accuracy.
Vision-based methods \cite{wang2025timemixer,wu2023timesnet} often provide limited semantic interpretability, which makes it difficult to incorporate domain knowledge or task-specific priors.
This motivates a vision-language (VL) {forecasting model} that combines visual feature extraction with LLM sequence modeling to improve short-term attitude forecasting for proactive beam steering.
Moreover, proactive attitude-aware beam steering must update the analog beam alignment for the next few time slots. 
{Meanwhile, the downlink still requires online multiuser digital beamforming and user admission to satisfy transmit-power and QoS constraints.
More broadly, recent robustness-oriented studies in Internet of Things (IoT) systems also reinforce the view that robustness should be treated as a primary design objective in dynamic networked systems\cite{chencuiot2026,chen2026lego,chen2025fast,chen2024distributed}.}

{These observations suggest that the core challenge is not forecasting alone, nor beamforming alone, but the conversion from imperfect attitude forecasts to robust and feasible hybrid beamforming actions. 
This motivates our forecast-then-beamform framework. 
Specifically, a VL forecasting model generates short-term attitude forecasts from multivariate flight telemetry, offline calibration converts forecast errors into analog-pointing robustness margins, and a QoS-driven digital beamforming and admission stage ensures instantaneous online feasibility.}
The main contributions of this paper are summarized as follows:
\begin{itemize}
\item {\textbf{Delay-aware forecast-then-beamform framework for HAP downlinks.}
We propose a multimodal framework that couples short-term HAP attitude forecasting with robust hybrid beamforming. 
The forecast attitudes provide timely input for proactive analog beam steering under practical forecast-to-actuation delay.}

\item {\textbf{Vision-language forecasting model for short-term attitude forecasting.}
We develop a VL-LLM forecasting model that combines vision-based feature extraction and LLM-based sequence modeling, leveraging their complementary strengths to improve multivariate HAP attitude forecasting for proactive beam steering.}

\item {\textbf{Offline calibration and QoS-driven beamforming with lightweight feasibility enforcement.}
We design an offline calibration procedure that estimates a high-confidence bound on the attitude forecasting error from an independent validation set.
The calibrated bound is used to account for pointing uncertainty in analog steering, while the proposed QoS-driven digital beamforming and admission stage ensures instantaneous transmit-power and user-rate feasibility online.}

\item {
Finally, we compare the proposed methods with representative forecasting and beamforming baselines and conduct comprehensive ablation and robustness studies. 
Simulation results show that the proposed framework improves both attitude forecasting accuracy and communication service performance over representative baselines.}
\end{itemize}

$Notation$: 
Bold uppercase and lowercase letters (for example, $\mathbf{A}$ and $\mathbf{x}$) denote matrices and column vectors, respectively. 
Plain letters denote scalars.
The imaginary unit is $j=\sqrt{-1}$.
$\mathbb{R}^{p\times q}$ and $\mathbb{C}^{p\times q}$ denote the sets of real and complex matrices.
$(\cdot)^{\top}$ and $(\cdot)^{\mathrm{H}}$ denote transpose and Hermitian transpose.
$|\cdot|$ denotes the absolute value, $\|\cdot\|_2$ the Euclidean norm, and $\|\cdot\|_{F}$ the Frobenius norm.
$\mathbb{E}\{\cdot\}$ denotes expectation and $\triangleq$ denotes definition.
$\otimes$ denotes the Kronecker product and $\mathrm{diag}(\cdot)$ forms a diagonal matrix.
For vectors, $\min(\cdot,\cdot)$ and $\max(\cdot,\cdot)$ operate elementwise, and $[\cdot]_+$ denotes elementwise projection onto $\mathbb{R}_+$.
$\mathcal{CN}(\mu,\sigma^2)$ denotes a circularly symmetric complex Gaussian distribution with mean $\mu$ and variance $\sigma^2$.
$\operatorname{atan2}(\cdot,\cdot)$ denotes the four-quadrant arctangent.
For a rotation matrix $\mathbf{R}\in\mathrm{SO}(3)$, $\log(\mathbf{R})\in\mathfrak{so}(3)$ denotes the matrix logarithm mapping to the Lie algebra of $\mathrm{SO}(3)$, and $\mathrm{vee}(\cdot):\mathfrak{so}(3)\rightarrow\mathbb{R}^3$ maps a skew-symmetric matrix to its corresponding vector representation.
Frequently used paper-specific symbols are summarized in Table~\ref{tab:key_symbols}.

\begin{table}[t]
\caption{{Key Symbols Used in the Proposed Framework}}
\label{tab:key_symbols}
\centering
\begin{tabular}{ll}
\hline
{Symbol} & {Meaning} \\
\hline
{$\Delta T$} & {Slot duration} \\
{$T_{\mathrm{lat}}$} & {End-to-end processing latency} \\
{$d$} & {Decision delay} \\
{$H_{\mathrm{pred}}$} & {Forecast horizon} \\
{$\mathcal{T}_t$} & {Target-slot set} \\
{$\mathbf{A}_{t+h|t}$} & {Forecast-based analog beamformer scheduled for slot $t+h$} \\
{$\mathbf{A}_{\tau}$} & {Analog beamformer actually applied in target-slot $\tau$} \\
{$\Delta\boldsymbol\omega_{\tau}$} & {Pointing mismatch induced by forecasting error} \\
{$\delta_{\omega}$} & {Offline-calibrated pointing-error bound} \\
{$\rho,\rho_s$} & {Confidence parameters in the certificates} \\
{$\alpha_k$} & {Admission indicator of user $k$} \\
{$r_k^{\min}$} & {Minimum rate requirement of user $k$} \\
{QAR} & {QoS admission ratio} \\
{$N_{\mathrm{RF}}$} & {Number of RF chains} \\
\hline
\end{tabular}
\end{table}

\section{System Model and Problem Formulation}\label{sec:system}

\subsection{System Model}\label{subsec:sys_model}

\begin{figure*}[t]
\centering
\begin{minipage}{0.4\linewidth}
  \centering
  \includegraphics[width=\linewidth]{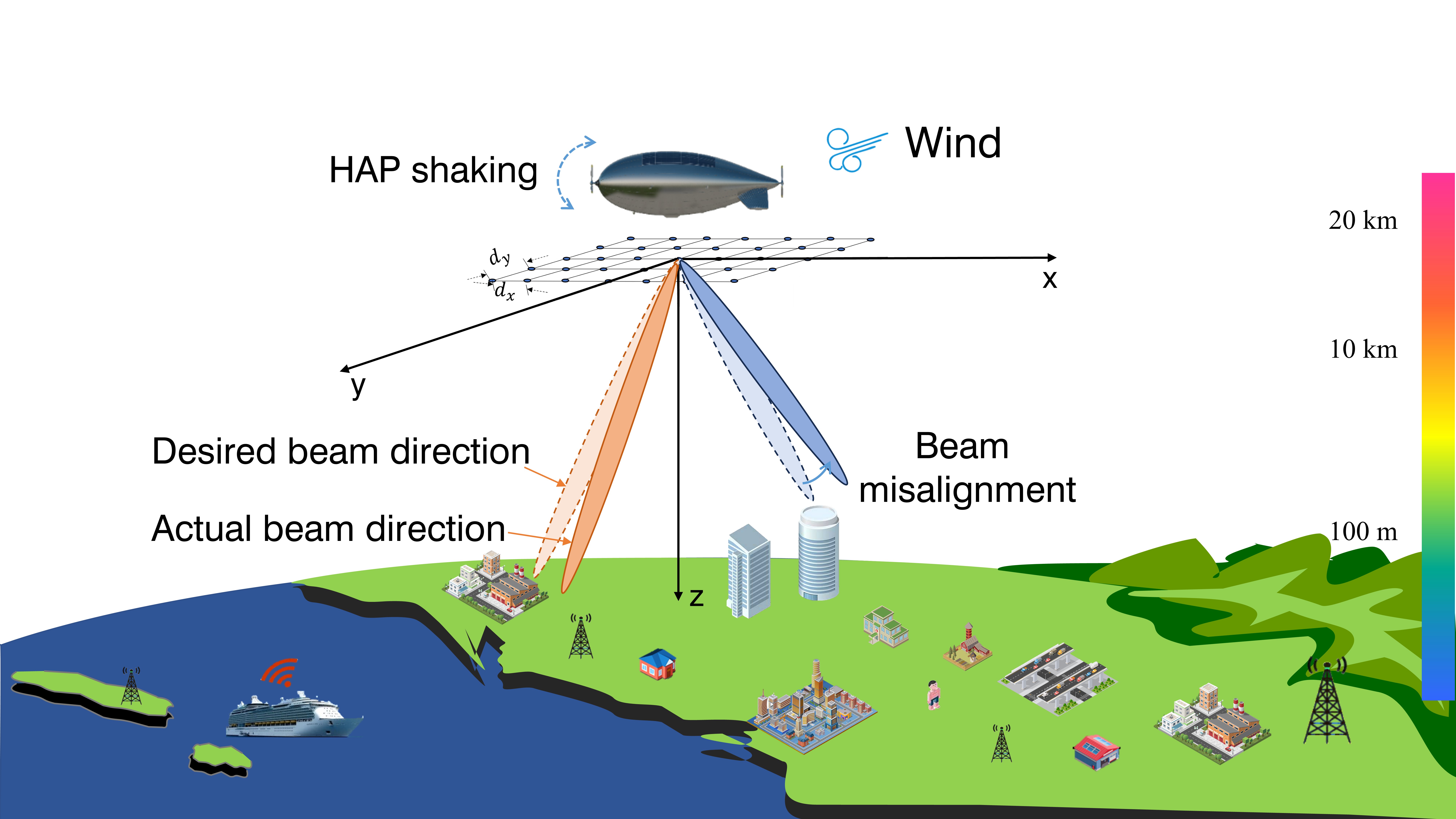}
  \caption{Downlink beam misalignment for a HAP under wind-induced attitude {shaking}. The HAP carries an $M_x \!\times\! M_y$ UPA. Dashed lobes denote desired beams, and solid lobes denote actual beams under attitude errors.}
  \label{fig:beam_misalignment}
\end{minipage} \qquad \qquad 
\begin{minipage}{0.4\linewidth}
  \centering
  \includegraphics[width=\linewidth]{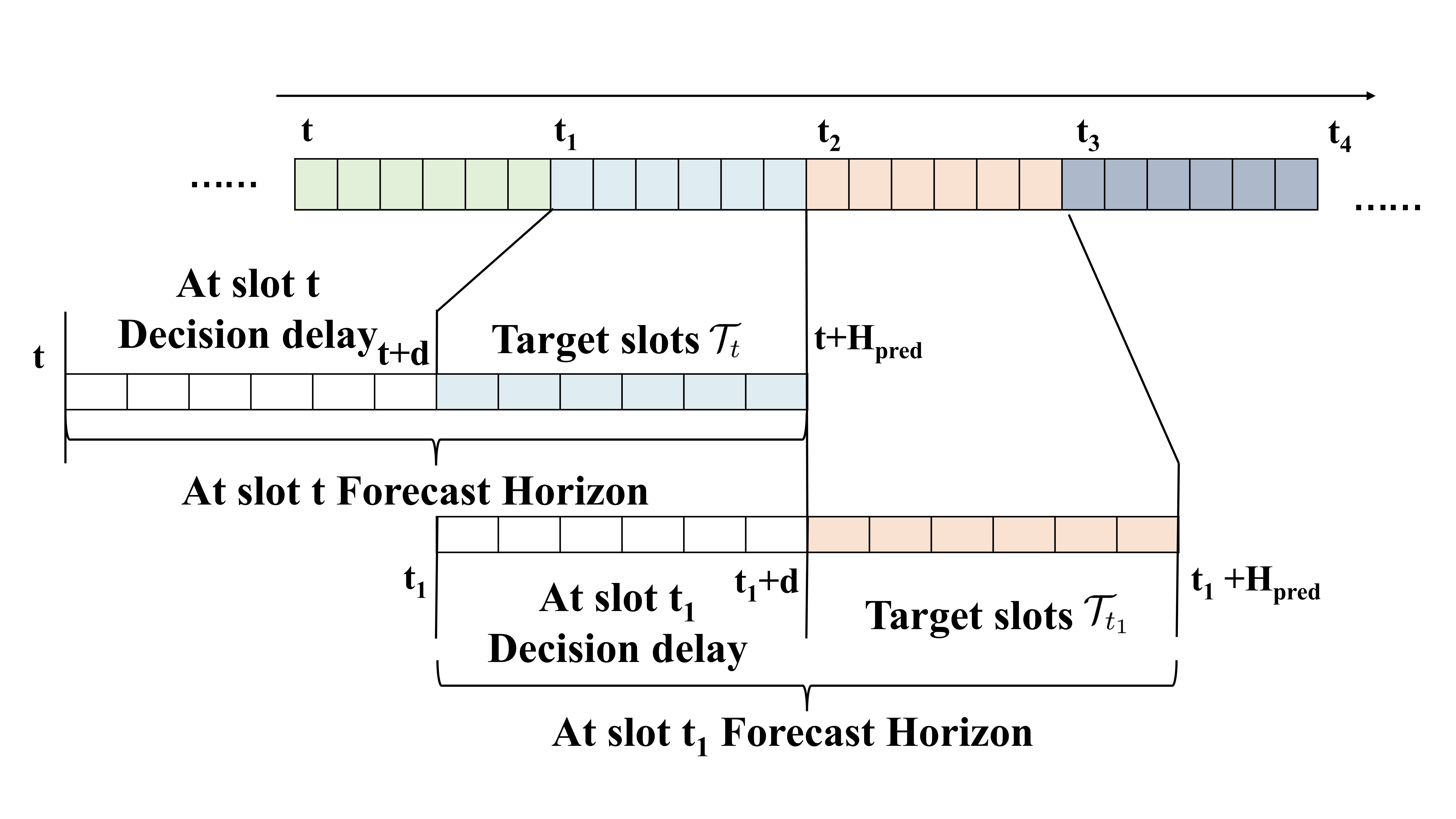}
  \caption{{Forecast} horizon with decision delay $d$ and target-slot set $\mathcal{T}_t=\{t{+}d{+}1,\ldots,t{+}H_{\mathrm{pred}}\}$. The figure shows the case $H_{\mathrm{pred}}=2d$.}
  \label{fig:horizon}
\end{minipage}
\end{figure*}

As shown in Fig.~\ref{fig:beam_misalignment}, we consider an HAP downlink communication system where an HAP equipped with a $M_x\times M_y$ UPA ($M\!=\!M_xM_y$) serves $K$ single-antenna ground users indexed by $\mathcal{K}=\{1,\ldots,K\}$.
The HAP is located at $\mathbf{r}_{\mathrm{HAP}}=(x_{\mathrm{HAP}},y_{\mathrm{HAP}},h_{\mathrm{HAP}})$ and the user $k$ is located at $\mathbf{r}_k=(x_k,y_k,0)$. 
{However, stratospheric winds may perturb the HAP attitude, leading to misalignment between the desired beam pointing and the actual UPA beam.}
We use a slotted time model and assume that the HAP attitude remains constant within each slot.

\subsubsection{{Forecast} horizon model}
As shown in Fig.~\ref{fig:horizon}, we consider a slotted timeline. The figure shows successive forecast horizons generated at slot $t$ and at the subsequent slot.
At the beginning of slot $t$, the {forecasting model} outputs an $H_{\mathrm{pred}}$-step attitude forecast.
Due to the end-to-end processing time $T_{\mathrm{lat}}$, the updated beamformer cannot take effect immediately.
We define the resulting decision delay as
\begin{equation}
d \triangleq \left\lceil \frac{T_{\mathrm{lat}}}{\Delta T} \right\rceil ,
\label{eq:decision_delay_def}
\end{equation}
where $\Delta T$ is the slot duration.
Therefore, only the forecasts for horizons $h\in\{d{+}1,\ldots,H_{\mathrm{pred}}\}$ are available for computing the analog beamformer sequence.
Accordingly, we define the target-slot set as
\begin{equation}
\mathcal{T}_t \triangleq \{t{+}d{+}1,\ldots,t{+}H_{\mathrm{pred}}\},\qquad |\mathcal{T}_t|=H_{\mathrm{pred}}-d.
\label{eq:target_slot_set}
\end{equation}
which includes the future slots within the {forecast} horizon that remain available after the decision delay.

We then compute an analog beamformer sequence $\{\mathbf{A}_{t+h|t}\}_{h=d+1}^{H_{\mathrm{pred}}}$ for $\tau=t+h\in\mathcal{T}_t$.
If a slot $\tau$ belongs to multiple target-slot sets $\mathcal{T}_{t'}$, we use the beamformer computed at the most recent {forecast} time, i.e.,
\begin{equation}
\mathbf{A}_\tau=\mathbf{A}_{\tau|t^\star(\tau)},\qquad
t^\star(\tau)\triangleq \max\{t':\,\tau\in\mathcal{T}_{t'}\}.
\label{eq:latest_cover_rule}
\end{equation}

\subsubsection{{Hybrid beamforming signal model}}
Hybrid beamforming is widely adopted to reduce the number of RF chains and the associated cost and power consumption \cite{6717211,7389996}.
We adopt an analog beamformer $\mathbf{A}\in\mathbb{C}^{M\times N_{\mathrm{RF}}}$ and a digital beamformer $\mathbf{D}\in\mathbb{C}^{N_{\mathrm{RF}}\times K}$.
At the beginning of slot $t$, we compute an analog beamformer sequence for $\mathcal{T}_t$ based on the HAP attitude {forecasts}, and then optimize $\mathbf{D}$ online in each target-slot using the instantaneous effective channel $\mathbf{H}_{\mathrm{eff}}(\tau)$ after applying the analog beamformer.
For clarity, we set $N_{\mathrm{RF}}=K$,
and the case $N_{\mathrm{RF}}<K$ can be handled by user scheduling/grouping and serving fewer data streams.
The transmitted signal is $\mathbf{x}(t)=\mathbf{A}\mathbf{D}\mathbf{s}(t)$ with $\mathbb{E}\{\mathbf{s}(t)\mathbf{s}^{\mathrm H}(t)\}=\mathbf{I}_K$.
Let $\mathbf{D}=[\mathbf{d}_1,\ldots,\mathbf{d}_K]$, where $\mathbf{d}_k$ is the digital beamformer for user $k$.
The instantaneous transmit power satisfies $\|\mathbf{A}\mathbf{D}\|_F^2\le P_{\max}$, and a constant circuit power $P_c$ is included for energy efficiency (EE) evaluation.

Let $\mathbf{h}_k(t)\in\mathbb{C}^{M}$ denote the downlink channel from the HAP array to user $k$ in slot $t$, and define
$\mathbf{H}(t)\triangleq [\mathbf{h}_1(t),\ldots,\mathbf{h}_K(t)]\in\mathbb{C}^{M\times K}$.
The received signal at user $k$ is
\begin{equation}
y_k(t)=\mathbf{h}_k^{\mathrm H}(t)\mathbf{A}\mathbf{D}\mathbf{s}(t)+n_k(t),
\end{equation}
where $n_k(t)\sim\mathcal{CN}(0,\sigma^2)$ denotes noise.
The instantaneous SINR of user $k$ is
\begin{equation}
\mathrm{SINR}_k(t)
=
\frac{\big|\mathbf{h}_k^{\mathrm H}(t)\mathbf{A}\mathbf{d}_k\big|^2}
{\sum_{j\neq k}\big|\mathbf{h}_k^{\mathrm H}(t)\mathbf{A}\mathbf{d}_j\big|^2 + \sigma^2 } ,
\label{eq:sinr_def}
\end{equation}
and the achievable rate is
\begin{equation}
R_k(t)=B\log_2\!\big(1+\mathrm{SINR}_k(t)\big),
\label{eq:sinr_rate}
\end{equation}
where $B$ is the system bandwidth.

\subsubsection{Modeling of impact of HAP shaking}
Let $\mathcal{W}$ and $\mathcal{B}$ denote the world frame and body frame, respectively.
The ZYX attitude at time $t$ is $(\psi_t,\theta_t,\phi_t)$ and the corresponding rotation matrices are \cite{11045512}
\begin{align}
\mathbf{R}_{\mathcal{B}\to\mathcal{W}}(t)&=\mathbf{R}_z(\psi_t)\,\mathbf{R}_y(\theta_t)\,\mathbf{R}_x(\phi_t),\notag\\
\mathbf{R}_{\mathcal{W}\to\mathcal{B}}(t)&=\mathbf{R}_{\mathcal{B}\to\mathcal{W}}^{\top}(t).
\label{eq:Rbw}
\end{align}

We represent the attitude by a rotation matrix $\mathbf{R}(t)\in \mathrm{SO}(3)$, which avoids the discontinuities and coupling issues of Euler-angle parameterizations \cite{geist2023rotations}.

Over the short forecast horizon, we assume a quasi-static geometry. 
Ground users are stationary.
Within each slot, the LoS direction and distance are treated as constant and are given by

\begin{equation}
\mathbf{e}_k \triangleq \frac{\mathbf{r}_k-\mathbf{r}_{\mathrm HAP}}{\|\mathbf{r}_k-\mathbf{r}_{\mathrm HAP}\|},\qquad
d_k \triangleq \|\mathbf{r}_k-\mathbf{r}_{\mathrm HAP}\|.
\label{eq:los_world_def}
\end{equation}

{Accordingly, we model the channel of user $k$ in slot $t$ as
\begin{equation}
\mathbf{h}_k(t)=
\sqrt{\frac{\kappa_k}{\kappa_k+1}}\,\mathbf{h}^{\mathrm{LoS}}_k(t)
+
\sqrt{\frac{1}{\kappa_k+1}}\,\mathbf{h}^{\mathrm{NLoS}}_k(t),
\label{eq:rician_channel}
\end{equation}
where $\kappa_k$ is the Rician factor, $\mathbf{h}^{\mathrm{LoS}}_k(t)$ is determined by the attitude-dependent steering geometry, and $\mathbf{h}^{\mathrm{NLoS}}_k(t)$ captures the diffuse component.}

We next specify how the attitude {forecast} is translated into phase-only analog beamforming in each target-slot.
Consider a target-slot $\tau=t+h$ with $h\in\{d{+}1,\ldots,H_{\mathrm{pred}}\}$.
Given the {forecast} attitude $(\widehat{\psi}_{t+h|t},\widehat{\theta}_{t+h|t},\widehat{\phi}_{t+h|t})$, we map the LoS direction of user $k$ into the {corresponding} body frame as \cite{6717211,7389996}
\begin{equation}
\mathbf{u}^{(\mathcal{B})}_{k,t+h|t}
=
\mathbf{R}_{\mathcal{W}\to\mathcal{B}}(\widehat{\psi}_{t+h|t},\widehat{\theta}_{t+h|t},\widehat{\phi}_{t+h|t})\,
\mathbf{e}_k,
\label{eq:pred_body_dir}
\end{equation}
where $\mathbf{u}^{(\mathcal{B})}_{k,t+h|t}=[u_x,u_y,u_z]^{\top}$.
The corresponding steering angles are
\begin{equation}
(\widetilde{\vartheta}_{k,t+h|t},\widetilde{\varphi}_{k,t+h|t})
=
\big(\arccos(u_z),\,\operatorname{atan2}(u_y,u_x)\big).
\label{eq:pred_body_angles}
\end{equation}

For the UPA, we use inter-element spacings $(d_x,d_y)$ and index elements by $(m,n)$ with $m=0,\ldots,M_x-1$ and $n=0,\ldots,M_y-1$.
The phase-only analog beam for user $k$ is implemented by per-element phase shifters following standard hybrid beamforming and array-response constructions \cite{6717211,7389996} as
\begin{equation}
\label{eq:upa_steering_vector}
\begin{split}
\big[\mathbf{w}_{k,t+h|t}\big]_{m,n}
&=
\frac{1}{\sqrt{M}}
\exp\!\Bigg(
j\frac{2\pi}{\lambda_c}\Big(
m d_x\sin\widetilde{\vartheta}_{k,t+h|t}\cos\widetilde{\varphi}_{k,t+h|t} \\
&\qquad\qquad
+
n d_y\sin\widetilde{\vartheta}_{k,t+h|t}\sin\widetilde{\varphi}_{k,t+h|t}
\Big)\Bigg).
\end{split}
\end{equation}
where $\lambda_c$ is the carrier wavelength.
With $N_{\mathrm{RF}}$ RF chains, the analog beamformer used in target-slot $t+h$ is formed by stacking user beams \cite{6717211,7389996} as
\begin{align}
\mathbf{A}_{t+h|t}
=
&\big[\mathbf{w}_{1,t+h|t},\ldots,\mathbf{w}_{K,t+h|t}\big]
\in\mathbb{C}^{M\times N_{\mathrm{RF}}},\notag\\
&\qquad
|[\mathbf{A}_{t+h|t}]_{p,q}|=\frac{1}{\sqrt{M}}.
\label{eq:pred_analog_beamformer_matrix}
\end{align}

During online operation, $\{\mathbf{A}_{t+h|t}\}_{h=d+1}^{H_{\mathrm{pred}}}$ is applied across the target-slots in $\mathcal{T}_t$ according to \eqref{eq:latest_cover_rule}, while the digital beamformer is updated per target-slot using the instantaneous channel estimate.

For a target-slot $\tau$, the analog beamformer is fixed as $\mathbf{A}_{\tau|t^\star(\tau)}$.
At the beginning of slot $\tau$, the HAP estimates the instantaneous channel and forms
\begin{equation}
\mathbf{H}_{\mathrm{eff}}(\tau)=\mathbf{H}^{\mathrm H}(\tau)\mathbf{A}_{\tau|t^\star(\tau)}.
\label{eq:Heff_def}
\end{equation}

Based on $\mathbf{H}_{\mathrm{eff}}(\tau)$, the digital beamformer is computed online under the instantaneous power and rate constraints.

In target-slot $\tau$, denote the {forecast attitude} by $\widehat{\mathbf{R}}_{\tau|t^\star(\tau)}$ and the realized attitude by $\mathbf{R}_\tau$.
We capture the overall pointing error relevant to analog beamforming by the relative rotation
\begin{equation}
\Delta\boldsymbol\omega_{\tau}
\triangleq
\mathrm{vee}\!\Big(\log\big(\widehat{\mathbf{R}}_{\tau|t^\star(\tau)}^{\top}\mathbf{R}_\tau\big)\Big)\in\mathbb{R}^3,
\label{eq:delta_omega_def}
\end{equation}
which captures the pointing mismatch caused by attitude {forecast} errors.
In online digital beamforming, $\Delta\boldsymbol\omega_{\tau}$ is treated as an uncertainty. 
In offline calibration, it is computed from ground-truth attitudes.

We focus on the pointing mismatch caused by attitude forecast errors.
For user $k$, the normalized detuning vector $\boldsymbol\xi_k=[\xi_{x,k},\xi_{y,k}]^{\top}$ admits a first-order model
\begin{equation}
\begin{bmatrix}\xi_{x,k}\\ \xi_{y,k}\end{bmatrix}
\approx
\mathbf{J}_{\xi,k}\,\Delta\boldsymbol\omega_{\tau},
\label{eq:Jxi_closed}
\end{equation}
where $\boldsymbol\xi_k$ is a 2D detuning in the direction-cosine domain and $\mathbf{J}_{\xi,k}\in\mathbb{R}^{2\times 3}$ is the corresponding Jacobian evaluated at the operating point associated with $\mathbf{A}_{\tau|t^\star(\tau)}$.

Under small detuning, the normalized array factor yields a quadratic main-lobe gain loss
\begin{align}
&\Delta G_{\mathrm{UPA},k}\approx c_x\,\xi_{x,k}^2+c_y\,\xi_{y,k}^2,\notag\\
\quad
&c_x=\tfrac{\pi^2}{3}(M_x^2-1),\ \ c_y=\tfrac{\pi^2}{3}(M_y^2-1).
\label{eq:loss_quad}
\end{align}

If $\Delta\boldsymbol\omega_{\tau}$ has mean $\boldsymbol\mu_\omega$ and covariance $\boldsymbol\Sigma_\omega$, then
\begin{equation}
\mathbb{E}\!\begin{bmatrix}\xi_{x,k}\\ \xi_{y,k}\end{bmatrix}
=\mathbf{J}_{\xi,k}\,\boldsymbol\mu_{\omega},\quad
\boldsymbol\Sigma_{\xi,k}
=\mathbf{J}_{\xi,k}\,\boldsymbol\Sigma_{\omega}\,\mathbf{J}_{\xi,k}^{\top}.
\label{eq:xi_stats_closed}
\end{equation}

We further define a scalar uncertainty proxy
\begin{equation}
\widehat{\sigma}^2_{\xi,k}\triangleq \mathrm{tr}\!\big(\boldsymbol\Sigma_{\xi,k}\big),
\label{eq:sigma_xi_scalar}
\end{equation}
which summarizes the total detuning variance along the two cosine-domain. 
These error measures will be used to impose robustness constraints in the problem formulation below.

\subsection{Problem Formulation}\label{subsec:pf}

We focus on a snapshot at the target-slot $\tau$.
When serving all users is infeasible under the instantaneous power budget $P_{\max}$, we introduce binary admission indicators $\alpha_k\in\{0,1\}$, where $\alpha_k=1$ indicates that user $k$ is admitted and must satisfy its QoS target.
The QoS admission ratio (QAR) is defined as
\begin{equation}
\mathrm{QAR}\triangleq \frac{1}{K}\sum_{k=1}^{K} \alpha_k .
\label{eq:QAR_def}
\end{equation}

\subsubsection{Robustness modeling and offline-calibrated certificates} Such pointing errors can reduce the main-lobe gain of narrow UPA beams.
We describe $\Delta\boldsymbol\omega_{\tau}$ using an offline-calibrated error model and use it to derive tractable robustness constraints.
We consider a bounded-error model and a moment-based model.
We introduce a design tolerance $\epsilon>0$ that specifies the maximum tolerable normalized main-lobe gain loss due to pointing errors.

\textit{Assumption A1 (Offline-calibrated target-window pointing bound).}
There exist a confidence level $1-\rho$ with $\rho\in(0,1)$ and a radius $\delta_\omega>0$ such that, for each {forecast time} $t$,
\begin{equation}
\mathbb{P}\!\left(
\max_{h\in\{d{+}1,\ldots,H_{\mathrm{pred}}\}}
\big\|\Delta\boldsymbol\omega_{t+h|t}\big\|_2
\le \delta_\omega
\right)\ge 1-\rho,
\label{eq:A1_window}
\end{equation}
where $\Delta\boldsymbol\omega_{t+h|t}$ is the pointing mismatch induced by using the $h$-step-ahead {forecast} generated at time $t$,
\begin{align}
\Delta\boldsymbol\omega_{t+h|t}
\triangleq
\mathrm{vee}\!\Big(\log\big(\widehat{\mathbf{R}}_{t+h|t}^{\top}\mathbf{R}_{t+h}\big)\Big)\in\mathbb{R}^3,\notag\\
\qquad h\in\{d{+}1,\ldots,H_{\mathrm{pred}}\}.
\label{eq:delta_omega_ts}
\end{align}

The radius $\delta_\omega$ is obtained offline on an independent validation/calibration split by taking an empirical $(1-\rho)$-quantile of the target-window maxima, following standard forecasting practice and recent conformal methods for multi-step time-series forecasting \cite{SOUSA2024128434}.

\textit{Assumption A2 (Effective angular region).} 
Let $(\vartheta_k,\varphi_k)$ denote the {forecast} body-frame steering angles used to form the analog beam in the target-slot, i.e.,
$(\vartheta_k,\varphi_k)\triangleq(\widetilde{\vartheta}_{k,\tau|t^\star(\tau)},\widetilde{\varphi}_{k,\tau|t^\star(\tau)})$.
For each user $k$, there exists an angular region $\mathcal{S}_k$ such that
\begin{equation}
\mathbb{P}\!\big((\vartheta_k,\varphi_k)\in\mathcal{S}_k\big)\ge 1-\rho_s,
\label{eq:A2}
\end{equation}
where $\rho_s\in(0,1)$ is a small confidence parameter.
The region $\mathcal{S}_k$ is specified offline from the service geometry and validation data, consistent with sector-based descriptions used in aerial beamforming~\cite{10436051}.

For small detuning, the main-lobe gain loss can be approximated by the quadratic model in \eqref{eq:loss_quad}.
Combining \eqref{eq:Jxi_closed} and \eqref{eq:loss_quad}, we can write
\begin{align}
&\Delta G_{\mathrm{UPA},k}
\approx
\Delta\boldsymbol\omega_{\tau}^{\top}\,
\mathbf{Q}_k(\vartheta_k,\varphi_k)\,
\Delta\boldsymbol\omega_{\tau},\notag\\
\qquad&
\mathbf{Q}_k(\vartheta,\varphi)
\triangleq
\mathbf{J}_{\xi,k}(\vartheta,\varphi)^{\top}
\mathrm{diag}(c_x,c_y)
\mathbf{J}_{\xi,k}(\vartheta,\varphi).
\label{eq:Qk_form}
\end{align}

To obtain a conservative constant that is valid over the angular region $\mathcal{S}_k$, we take a worst-case spectral bound
\begin{equation}
L_k^2 \triangleq \max_{(\vartheta,\varphi)\in\mathcal{S}_k}\lambda_{\max}\!\big(\mathbf{Q}_k(\vartheta,\varphi)\big).
\label{eq:Lk_def_rev}
\end{equation}

\begin{lemma}\label{lem:TI_certificate}
\textbf{Offline-calibrated pointing-feasibility certificate.}
Under Assumptions~A1-A2, define $L_k^2$ by \eqref{eq:Lk_def_rev}.
On the event
$\left\{\max_{h\in\{d{+}1,\ldots,H_{\mathrm{pred}}\}}\|\Delta\boldsymbol\omega_{t+h|t}\|_2\le \delta_\omega\right\}
\cap \{(\vartheta_k,\varphi_k)\in\mathcal{S}_k\}$,
if $L_k^2\,\delta_\omega^2\le \epsilon$, then the quadratic approximation of the main-lobe gain loss satisfies
$\Delta G_{\mathrm{UPA},k}\le \epsilon$
for user $k$ over all target-slot horizons $h\in\{d{+}1,\ldots,H_{\mathrm{pred}}\}$ programmed at time $t$.
Consequently, the certificate holds with probability at least $1-(\rho+\rho_s)$.
\end{lemma}
\begin{proof}
    Please refer to Appendix 1 in the supplementary material.
\end{proof}

Lemma~\ref{lem:TI_certificate} provides an offline-calibrated pointing-feasibility certificate.
In particular, enforcing $L_k^2\delta_\omega^2\le \epsilon$ guarantees $\Delta G_{\mathrm{UPA},k}\le \epsilon$ on the calibrated target-window event, with probability at least $1-(\rho+\rho_s)$.
This motivates the sufficient robustness constraint in \eqref{eq:PF_TI_cert}, which introduces no additional online decision variables once $\mathbf{A}_\tau$ is fixed.

{In addition to the bounded-error certificate above, we also consider a moment-based condition.}
If $\Delta\boldsymbol\omega_{\tau}$ has mean $\boldsymbol\mu_{\omega}$ and covariance $\boldsymbol\Sigma_{\omega}$, then
\begin{equation}
\mathbb{E}\!\left[\Delta G_{\mathrm{UPA},k}\right]
=\boldsymbol\mu_{\omega}^{\top}\mathbf{Q}_k(\vartheta_k,\varphi_k)\boldsymbol\mu_{\omega}
+\mathrm{tr}\!\left(\mathbf{Q}_k(\vartheta_k,\varphi_k)\boldsymbol\Sigma_{\omega}\right)
\le \epsilon,
\label{eq:stat_cert}
\end{equation}
which follows from standard second-moment analysis~\cite{11048751}.

\subsubsection{QoS-driven two-stage optimization}

We adopt a QoS-driven policy. 
We first maximize QAR and then refine the EE on the admitted set.
Combining QoS, power, the analog constant-modulus structure, and the above certificate, the snapshot constraints are written as
\begin{subequations}\label{eq:PF_TI}
\begin{align}
& R_k(\mathbf{A}_\tau,\mathbf{D}) \ \ge\ \alpha_k\, r_k^{\min}, \ \forall k, \label{eq:PF_TI_qos}\\
& \|\mathbf{A}_\tau\mathbf{D}\|_F^2 \ \le\ P_{\max}, \label{eq:PF_TI_pow}\\
& \mathbf{A}_\tau\in\mathcal{A}^{\mathrm{pred}}
=\Big\{\mathbf{A}_\tau\in\mathbb{C}^{M\times N_{\mathrm{RF}}}:|[\mathbf{A}_\tau]_{p,q}|=\tfrac{1}{\sqrt{M}}, \notag\\
&\quad\quad\quad\quad\quad\quad\ \forall p=1,\ldots,M,\ q=1,\ldots,N_{\mathrm{RF}}\Big\},
\label{eq:PF_TI_CM}
\\
& \alpha_k\, L_k^2\,\delta_\omega^2 \ \le\ \alpha_k\,\epsilon,\ \forall k. \label{eq:PF_TI_cert}
\end{align}
\end{subequations}

Here, $\boldsymbol{\alpha}=[\alpha_1,\ldots,\alpha_K]^{\top}$ collects the binary admission indicators.
The analog beamformer $\mathbf{A}_\tau$ is generated from the {HAP attitude} via \eqref{eq:pred_body_dir}-\eqref{eq:pred_analog_beamformer_matrix} and therefore satisfies the constant-modulus constraint \eqref{eq:PF_TI_CM} by construction.
The certificate parameter $\delta_\omega$ is calibrated offline from {forecast errors} over the target-slot horizons $h\in\{d{+}1,\ldots,H_{\mathrm{pred}}\}$.

Given the analog beamformer $\mathbf{A}_\tau$ fixed for the target-slot $\tau$, we adopt a QoS-driven two-stage formulation.
Stage I maximizes QAR and is formulated as
\begin{subequations}\label{prob:QAR_static}
\begin{align}
\underset{\mathbf{D},\,\boldsymbol{\alpha}}{\text{maximize}}\quad &
\frac{1}{K}\sum_{k=1}^{K} \alpha_k \label{prob:QAR_static:a}\\
\text{subject to}\quad
& \eqref{eq:PF_TI_qos},\ \eqref{eq:PF_TI_pow},\ \eqref{eq:PF_TI_cert}, \label{prob:QAR_static:b}\\
& \alpha_k\in\{0,1\},\ \forall k. \label{prob:QAR_static:e}
\end{align}
\end{subequations}

Let $\boldsymbol{\alpha}^{\star}$ be an optimal admission decision of \eqref{prob:QAR_static}.
Stage II maximizes EE over the admitted set and is formulated as
\begin{subequations}\label{prob:EE_given_adm}
\begin{align}
\underset{\mathbf{D}}{\text{maximize}}\quad &
\frac{\sum_{k=1}^{K} \alpha_k^{\star}\,R_k(\mathbf{A}_\tau,\mathbf{D})}{\|\mathbf{A}_\tau\mathbf{D}\|_F^2+P_c}
\label{prob:EE_given_adm:a}\\
\text{subject to}\quad
& R_k(\mathbf{A}_\tau,\mathbf{D}) \ \ge\ \alpha_k^{\star}\, r_k^{\min},\ \forall k, \label{prob:EE_given_adm:b}\\
& \eqref{eq:PF_TI_pow}. 
\label{prob:EE_given_adm:c}
\end{align}
\end{subequations}

The resulting problems are non-convex and include binary admission variables.
Even with fixed $\mathbf{A}_\tau$, the SINR/rate constraints are non-convex due to multiuser interference coupling in \eqref{eq:sinr_def}.
Moreover, Stage I introduces binary admissions, leading to a mixed-integer non-convex program (MINCP), while Stage II features a fractional EE objective under instantaneous QoS and power constraints.
Such coupled formulations are widely recognized as challenging to solve at scale under stringent real-time processing constraints \cite{9839601,10517628}.
Beyond computational hardness, the analog beamformer must be set in advance based on the {forecast attitudes} over the horizon, whereas the digital beamformer must be computed online using instantaneous CSI.
Meanwhile, attitude forecast errors introduce pointing uncertainty, and robustness-aware design is needed without incurring high online complexity.
These gaps motivate the learning-assisted reconstruction and strict feasibility repair developed in Sec.~\ref{sec:algorithm}, while Sec.~\ref{sec:VL-LLM} specifies how VL-LLM produces $\{\mathbf{A}_{t+h|t}\}_{h=d+1}^{H_{\mathrm{pred}}}$ and how {forecast errors} over the target-slot window calibrate the certificate parameter used above.

\section{VL-LLM for {Forecast-Aided} Analog Beamforming}\label{sec:VL-LLM}

Recent time-series {forecasting} methods for multivariate HAP flight {telemetry} can be broadly categorized into vision-based methods \cite{wang2025timemixer,wu2023timesnet} and LLM-based methods \cite{10892257,jin2024timellm,10.1145/3589334.3645434}.
However, LLM-based methods face a modality gap between continuous-valued time series and discrete tokens, and they are not explicitly pre-trained for fine-grained temporal patterns, which can limit forecasting accuracy.
Vision-based methods often provide limited semantic interpretability, which makes it difficult to incorporate domain knowledge or task-specific priors.
These limitations motivate a vision-language {forecasting model} that combines visual feature extraction with LLM sequence modeling to improve short-term attitude forecasting for proactive analog beamforming under decision delay.
{To address these limitations, we combine vision-based feature extraction with prompt-based LLM conditioning.
The vision module renders multivariate flight telemetry as visual representations so that a vision backbone can capture local dynamics, cross-variable coupling, and periodic structures.
The language module acts as a prompt-conditioned sequence forecaster that injects domain-aware task priors into multi-step attitude forecasting.}

\begin{figure}[!t]
	\centering
	\includegraphics[width=\linewidth]{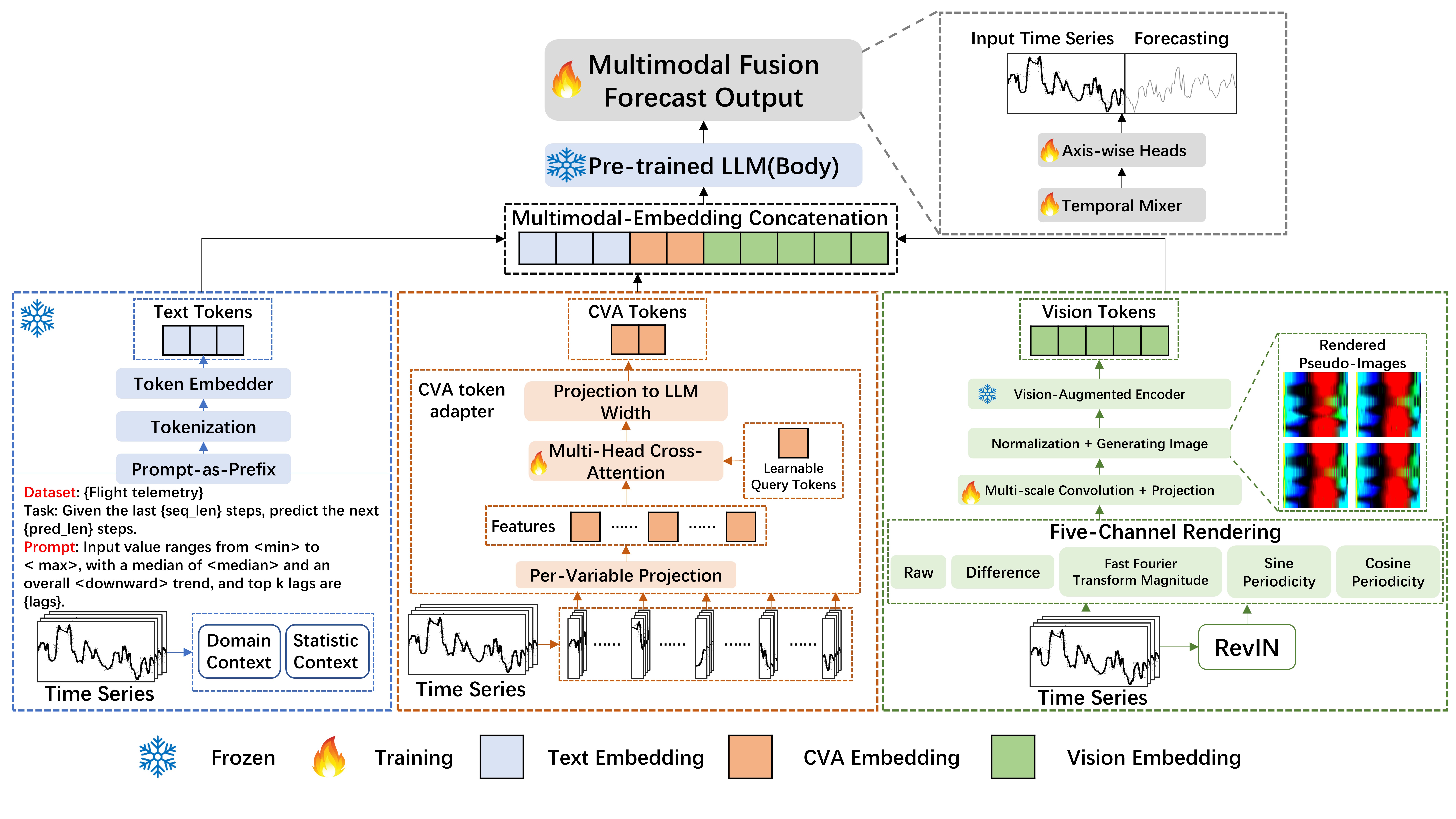}
	\caption{Overall VL-LLM method for {forecast-aided} analog beamforming in HAP downlinks.}
	\label{fig:VL-LLM}
\end{figure}

{As shown in Fig.~\ref{fig:VL-LLM}, the input to VL-LLM is a length-$L_{\mathrm{win}}$ multivariate telemetry window $\mat{X}_t\in\mathbb{R}^{N_{\mathrm{var}}\times L_{\mathrm{win}}}$. 
The model contains three complementary modules. 
First, a telemetry-to-vision module converts $\mat{X}_t$ into vision tokens for capturing local temporal dynamics, periodic structures, and frequency-sensitive patterns. 
Second, a small set of learned cross-variable adapter (CVA) tokens summarizes inter-channel coupling and provides a compact token-level representation of cross-variable dependencies. 
Third, a text module injects task instructions and lightweight window statistics into the frozen LLM. 
The resulting multimodal token sequence is then fed into a frozen LLM with lightweight adaptation.
This frozen LLM acts as a prompt-conditioned forecaster and outputs an $H_{\mathrm{pred}}$-step forecast of yaw, pitch, and roll. }
Due to the decision delay $d$, only the target-slot horizons $h\in\{d{+}1,\ldots,H_{\mathrm{pred}}\}$ are used to form the analog beamformer sequence for the target-slot set $\mathcal{T}_t$.
In each target-slot $\tau\in\mathcal{T}_t$, the digital beamformer is still optimized online from the instantaneous channel estimate.

\subsection{{Telemetry-to-Vision Module}}
{The vision module renders multivariate flight telemetry into a pseudo-image so that local transitions, cross-channel structures, and frequency-sensitive patterns can be encoded by a vision backbone.}
Given $\mat{X}_t\in\mathbb{R}^{N_{\mathrm{var}}\times L_{\mathrm{win}}}$, where $N_{\mathrm{var}}$ is the number of variables and $L_{\mathrm{win}}$ is the look-back window length, we apply channel-wise reversible instance normalization (RevIN) \cite{kim2022reversible} as
\begin{equation}
\widetilde{\mat{X}}_t=\mathrm{RevIN}(\mat{X}_t).
\label{eq:revin}
\end{equation}

We form first-order differences along the time axis.
We set $\Delta\mat{X}_t[:,1]=\mathbf{0}$ and define $\Delta\mat{X}_t[:,\ell]=\widetilde{\mat{X}}_t[:,\ell]-\widetilde{\mat{X}}_t[:,\ell-1]$ for $\ell=2,\ldots,L_{\mathrm{win}}$.
We compute the real fast Fourier transform (FFT) magnitude of each variable along the time axis and interpolate it to length $L_{\mathrm{win}}$.
\begin{equation}
\label{eq:fft-mag}
\overline{\mathcal{F}}(\widetilde{\mat{X}}_t)
=\mathrm{Interp}_{\to L_{\mathrm{win}}}\!\big(|\mathrm{rFFT}(\widetilde{\mat{X}}_t)|\big).
\end{equation}
where $\mathrm{rFFT}(\cdot)$ denotes the one-sided real FFT along the time axis,
and $\mathrm{Interp}_{\to L_{\mathrm{win}}}(\cdot)$ denotes interpolation that resamples the one-sided magnitude spectrum to length $L_{\mathrm{win}}$ for consistent stacking.

We form a multivariate tensor by concatenating raw values, first-order differences, FFT magnitudes, and sinusoidal codes along the channel axis.
{This process is} gated by $s_{\mathrm{raw}},s_{\Delta},s_{\mathrm{fft}},s_{\mathrm{per}}\in\{0,1\}$ {and is expressed} as
\begin{align}
\label{eq:vam-tensor}
\mathcal{X}^{\mathrm{vis}}_t
=&\mathrm{concat}_c\!\Big(
s_{\mathrm{raw}}\widetilde{\mat{X}}_t,\;
s_{\Delta}\Delta\mat{X}_t,\;\notag\\
&s_{\mathrm{fft}}\overline{\mathcal{F}}(\widetilde{\mat{X}}_t),\;
s_{\mathrm{per}}\widetilde{\mat{S}}_t,\;
s_{\mathrm{per}}\widetilde{\mat{C}}_t
\Big)
\in\mathbb{R}^{C_{\mathrm{in}}\times N_{\mathrm{var}}\times L_{\mathrm{win}}},
\end{align}
where $\mathrm{concat}_c(\cdot)$ denotes concatenation along the channel dimension and  $C_{\mathrm{in}}=s_{\mathrm{raw}}+s_{\Delta}+s_{\mathrm{fft}}+2s_{\mathrm{per}}$,
{$\widetilde{\mat{S}}_t$ and $\widetilde{\mat{C}}_t$ represent adaptive per-variable sinusoidal codes.}
{Raw values preserve short-term motion states, while first-order differences emphasize local transitions. 
FFT magnitudes summarize dominant frequency content, and sinusoidal codes provide lightweight periodicity cues.}
{These components are stacked to form a compact multi-channel visual tensor that integrates temporal, spectral, and periodic information.}
We convert $\mathcal{X}^{\mathrm{vis}}_t$ to a pseudo-RGB image using a lightweight convolutional stem {which} maps $C_{\mathrm{in}}$ channels to three channels for compatibility with image backbones.
We then resize the image to $(H_0,W_0)$ and apply the standard preprocessing required by the selected vision backbone.

\subsection{Cross-Variable Adapter Module}
\label{subsec:cva}
{While the telemetry-to-vision module is effective at extracting rendered local patterns, it does not explicitly provide a compact token-level summary of cross-variable coupling. To supplement this information, we introduce a lightweight CVA that aggregates inter-channel dependencies into a small set of learnable adapter tokens before the frozen LLM.}
{This module utilizes a small set of learnable query tokens and multi-head cross-attention over per-variable features to output a fixed number of adapter tokens.}
Given $\mat{X}_t\in\mathbb{R}^{N_{\mathrm{var}}\times L_{\mathrm{win}}}$, we take the most recent $L_{\mathrm{tail}}$ steps ($L_{\mathrm{tail}}\le L_{\mathrm{win}}$) and form per-variable tail patches $\widetilde{\mat{X}}^{\mathrm{tail}}_t[i,:]\in\mathbb{R}^{L_{\mathrm{tail}}}$ for $i=1,\ldots,N_{\mathrm{var}}$.
Each variable patch is projected to a latent width $d_{\mathrm{cva}}$ as
\begin{equation}
\label{eq:cva-var-proj}
\vect{v}_i = \mat{W}_{\mathrm{var}}\,\widetilde{\mat{X}}^{\mathrm{tail}}_t[i,:] + \vect{b}_{\mathrm{var}}
\in\mathbb{R}^{d_{\mathrm{cva}}},\quad i=1,\ldots,N_{\mathrm{var}}.
\end{equation}
where $\mat{W}_{\mathrm{var}}\in\mathbb{R}^{d_{\mathrm{cva}}\times L_{\mathrm{tail}}}$ and $\vect{b}_{\mathrm{var}}\in\mathbb{R}^{d_{\mathrm{cva}}}$ are shared learnable projection parameters.
We stack $\mat{V}_{\mathrm{var}}=[\vect{v}_1;\ldots;\vect{v}_{N_{\mathrm{var}}}]$ and use $N_{\mathrm{cv}}$ learnable queries to attend over variables via multi-head cross-attention as
\begin{equation}
\label{eq:cva-crossattn}
\mat{H}_{\mathrm{cva}} = \mathrm{MHA}(\mat{Q},\mat{V}_{\mathrm{var}},\mat{V}_{\mathrm{var}})\in\mathbb{R}^{N_{\mathrm{cv}}\times d_{\mathrm{cva}}}.
\end{equation}
where $\mat{Q}\in\mathbb{R}^{N_{\mathrm{cv}}\times d_{\mathrm{cva}}}$ denotes learnable query tokens, $\mat{V}_{\mathrm{var}}$ is used as both keys and values. 
{By employing a fixed number of learnable query tokens $N_{\mathrm{cv}}$, the adapter output length is independent of $N_{\mathrm{var}}$, which effectively bounds the additional computation introduced by the cross-variable module.}

A lightweight residual block refines the adapter tokens as
\begin{equation}
\label{eq:cva-block}
\mat{U}=\mathrm{LN}(\mat{Q}+\mat{H}_{\mathrm{cva}}),\quad
\mat{U}\leftarrow \mathrm{LN}\big(\mat{U}+\mathrm{FFN}(\mat{U})\big),
\end{equation}
where $\mathrm{LN}(\cdot)$ denotes layer normalization applied along the feature dimension, and $\mathrm{FFN}(\cdot)$ denotes a position-wise feed-forward network that preserves the token length and feature dimension.
We finally project them to the LLM width as
\begin{equation}
\label{eq:cva-out}
\widehat{\mat{M}}=\mat{U}\mat{W}_{\mathrm{out}} \in\mathbb{R}^{N_{\mathrm{cv}}\times d_{\mathrm{LLM}}}.
\end{equation}
where $\mat{W}_{\mathrm{out}}\in\mathbb{R}^{d_{\mathrm{cva}}\times d_{\mathrm{LLM}}}$ is a learnable linear projection that maps the CVA latent width $d_{\mathrm{cva}}$ to the LLM embedding width $d_{\mathrm{LLM}}$.
{Consequently, CVA produces $N_{\mathrm{cv}}$ tokens that are appended to the LLM input sequence.
Acting as a compact bridge between raw telemetry and the frozen backbone, these tokens encode inter-variable interactions.}

\subsection{Text Module}
{Alongside the telemetry-to-vision and CVA modules, we use a text module to inject task priors and lightweight window-level summaries into the frozen LLM. 
This module provides structured task conditioning for short-term HAP attitude forecasting rather than free-form language generation.}

\begin{table}[h]
\centering
\caption{{Representative prompt structure for the text module}}
\label{tab:prompt_structure}
\begin{tabular}{ll}
\toprule
{\textbf{Component}} & {\textbf{Content}} \\
\midrule
{Role} & {Expert HAP attitude forecaster} \\
{Input} & {Last \{seq\_len\} time steps} \\
{Task} & {Forecast next \{pred\_len\} steps for proactive beam steering} \\
{Output format} & {\{enc\_in\} real values per step, channel-ordered} \\
{Guidance} & {Temporal consistency, trends, periodic patterns} \\
{Context} & {Window statistics: slope and dominant period per channel} \\
\bottomrule
\end{tabular}
\end{table}
{Table~\ref{tab:prompt_structure} presents a representative prompt structure used in the text module.
The conditioning comprises high-level task priors, structured output constraints for numerical forecasting, and local context descriptors derived from the input window.}
We summarize trend statistics using a per-channel least-squares slope over the look-back window.
Let $\widetilde{X}_{t}[i,\ell]$ denote the RevIN-normalized value of channel $i$ at index
$\ell\in\{1,\ldots,L_{\mathrm{win}}\}$, and let
$\bar{\ell}=\frac{1}{L_{\mathrm{win}}}\sum_{\ell=1}^{L_{\mathrm{win}}}\ell$ and
$\bar{X}_{i}=\frac{1}{L_{\mathrm{win}}}\sum_{\ell=1}^{L_{\mathrm{win}}}\widetilde{X}_{t}[i,\ell]$.
{To instantiate the prompt in a lightweight and reproducible manner, we summarize each channel by extracting specific trend and periodicity descriptors from the current look-back window.} 
We compute the slope as
\begin{equation}
\label{eq:pap-slope}
\kappa_i
=\frac{\sum_{\ell=1}^{L_{\mathrm{win}}}(\ell-\bar{\ell})\big(\widetilde{X}_{t}[i,\ell]-\bar{X}_{i}\big)}
{\sum_{\ell=1}^{L_{\mathrm{win}}}(\ell-\bar{\ell})^2+\varepsilon_{\kappa}},
\quad i=1,\ldots,N_{\mathrm{var}},
\end{equation}
{where the variable $\kappa_i$ approximates the signed temporal trend per channel} and $\varepsilon_{\kappa}>0$ is a small constant for numerical stability.

To capture periodicity, we estimate a dominant period from the peak of the one-sided FFT power spectrum while excluding the zero-frequency component. 
{This is achieved through the following formulations}
\begin{align}
\label{eq:pap-fft-peak}
\mathcal{F}_i[f]=\mathrm{rFFT}\big(\widetilde{X}_t[i,1{:}L_{\mathrm{win}}]\big)[f],\notag\\\quad
f_i^\star=\arg\max_{f\ge 1}|\mathcal{F}_i[f]|^2,\quad
\widehat{P}_i \approx \frac{L_{\mathrm{win}}}{f_i^\star},
\end{align}
where $\mathcal{F}_i[f]$ denotes the one-sided FFT coefficient at frequency-bin index $f$ for channel $i$.
We additionally include dispersion summaries as
\begin{equation}
\label{eq:pap-disp}
\varsigma_i=\sqrt{\frac{1}{L_{\mathrm{win}}}\sum_{\ell=1}^{L_{\mathrm{win}}}\big(\widetilde{\mat{X}}_{t}[i,\ell]-\bar{X}_{i}\big)^2 }.
\end{equation}

Let $\Omega$ denote fixed domain tokens (HAP/beam task), $\mathcal{I}$ the instruction tokens, 
and $\mathcal{P}_t$ the textualization of $\{\{\kappa_i\}_{i=1}^{N_{\mathrm{var}}},\{\widehat{P}_i\}_{i=1}^{N_{\mathrm{var}}},\{\bar{X}_i,\varsigma_i\}_{i=1}^{N_{\mathrm{var}}}\}$. 
{A tokenizer and embedding map $\mathrm{Tok}(\cdot)$ converts this structured information into vectors according to}
\begin{equation}
\label{eq:pap-text}
\mat{E}_{\text{text}} \;=\; \mathrm{Tok}\!\big([\Omega;\,\mathcal{I};\,\mathcal{P}_t]\big)\in\mathbb{R}^{T_{\text{text}}\times d_{\mathrm{LLM}}}.
\end{equation}

\subsection{Multimodal Fusion and Output Heads}
{The multimodal fusion stage combines task-conditioned text tokens, CVA tokens, and vision tokens into a unified sequence for direct multi-horizon attitude forecasting.
We form the multimodal input sequence to the frozen LLM as}
\begin{equation}
\label{eq:mm-concat}
\mat{Z}_t
=\big[\,\mathrm{Trunc}_{T_{\max}}(\mat{E}_{\text{text}})\;;\;
\widehat{\mat{M}}\;;\;
\mathrm{Trunc}_{N_{\max}}(\mat{V}_t)\,\big],
\end{equation}
where $\mathrm{Trunc}_{L}(\cdot)$ keeps the first $L$ tokens along the token dimension.
The resulting multimodal sequence $\mat{Z}_t$ is fed to the frozen LLM, and the output representations are further processed by a lightweight mixer and axis-wise heads for multi-horizon forecasting.
{This concatenation preserves the complementary roles of the three modules where the text tokens provide task priors, the CVA tokens summarize cross-channel dependencies, and the vision tokens encode rendered telemetry patterns.}

{The last-layer hidden states of the frozen LLM are given by}
\begin{equation}
\mat{H}^{(L)}_{\mathrm{LLM}}
=
\big[\,\mat{H}^{(L)}_{\mathrm{LLM,txt}};\mat{H}^{(L)}_{\mathrm{LLM,cont}}\,\big]
\in\mathbb{R}^{(T_{\mathrm{txt}}+N_{\mathrm{cont}})\times d_{\mathrm{LLM}}},
\end{equation}
where $\mat{H}^{(L)}_{\mathrm{LLM,txt}}$ corresponds to the retained text tokens with length $T_{\mathrm{txt}}$ after truncation,
and $\mat{H}^{(L)}_{\mathrm{LLM,cont}}$ collects the remaining content tokens, including the $N_{\mathrm{cv}}$ CVA tokens and the $N_{\mathrm{vis}}$ vision tokens. 
{A lightweight temporal mixer $\Pi(\cdot)$ models the content-token sequence into a compact context vector $\vect{g}$ for direct multi-step forecasting, formulated as}
\begin{equation}
\label{eq:fbop-pool-revised}
\vect{g}=\Pi\!\big(\mat{H}^{(L)}_{\mathrm{LLM,cont}}\big).
\end{equation}
\begin{align}
\label{eq:fbop-direct-angles-revised}
\bigl[\widehat{\psi}_{t+1:t+H_{\mathrm{pred}}|t},\, \widehat{\theta}_{t+1:t+H_{\mathrm{pred}}|t},\, \widehat{\phi}_{t+1:t+H_{\mathrm{pred}}|t}\bigr]
= \notag\\
\bigl[g_{\text{yaw}}(\Gamma(\vect{g})),\, g_{\text{pitch}}(\Gamma(\vect{g})),\, g_{\text{roll}}(\Gamma(\vect{g}))\bigr],
\end{align}
where each axis-wise head outputs an $H_{\mathrm{pred}}$-dimensional sequence, and $\Gamma(\cdot)$ denotes an axis router/gating module.
Each head $g_{\mathrm{yaw}},g_{\mathrm{pitch}},g_{\mathrm{roll}}$ maps the context vector to an $H_{\mathrm{pred}}$-dimensional sequence, a direct multi-horizon forecasting that avoids error accumulation across steps.

The {forecast} attitudes are converted into the analog beamformer sequence {for the target-slot window} according to the geometric and steering relations in Section~\ref{subsec:sys_model}, while the digital beamformer is refined online using the instantaneous effective channel.

\subsection{Training and Offline Calibration}
{The forecasting model is trained for direct multi-horizon forecasting so that all future forecast steps are produced in one shot, which is consistent with the delay-aware target-slot beamformer framework.}
Given a length-$L_{\mathrm{win}}$ history ending at $t$, the model {forecasts} an $H_{\mathrm{pred}}$-step attitude sequence
\begin{equation}
\widehat{\mathbf{a}}_{t+s|t}\triangleq
(\widehat{\psi}_{t+s|t},\widehat{\theta}_{t+s|t},\widehat{\phi}_{t+s|t}),
\quad s=1,\ldots,H_{\mathrm{pred}}.
\end{equation}

Let $\mathbf{Y}_{t+1:t+H_{\mathrm{pred}}}\in\mathbb{R}^{H_{\mathrm{pred}}\times C}$ denote the ground-truth attitude sequence and
$\widehat{\mathbf{Y}}_{t+1:t+H_{\mathrm{pred}}}\in\mathbb{R}^{H_{\mathrm{pred}}\times C}$ denote the corresponding {forecast sequence}, where $C$ is the number of attitude channels.
The dataset standardizes the targets using scalers fitted on the training split, and we invert the scaling inside the loss so that errors are measured in the original physical units.
We minimize a weighted Huber regression loss
\begin{equation}
\mathcal{L}_{\mathrm{reg}}
=\frac{\sum_{s=1}^{H_{\mathrm{pred}}}\sum_{c=1}^{C} w_s\,w_c\,
\HuberFunc{\Delta_{s,c}}{\delta}}
{\sum_{s=1}^{H_{\mathrm{pred}}}\sum_{c=1}^{C} w_s\,w_c},
\end{equation}
where $\delta$ is the Huber threshold, $w_c$ are axis/channel weights, and $w_s$ are horizon weights.
The per-step error is
$\Delta_{s,c}=\widehat{\mathbf{Y}}_{t+1:t+H_{\mathrm{pred}}}[s,c]-\mathbf{Y}_{t+1:t+H_{\mathrm{pred}}}[s,c]$.
For yaw, we {use the sine-cosine representation} $(Y_{\sin},Y_{\cos})$ and evaluate the angular error via $\mathcal{L}_{\mathrm{yaw}}$ using $\mathrm{wrap}_{\pi}(\cdot)$ below.

We further add smoothness and geometry-aware terms and define the overall {forecasting} loss as
\begin{equation}
\mathcal{L}_{\mathrm{pred}}
=\mathcal{L}_{\mathrm{reg}}
+\lambda_{\mathrm{vel}}\mathcal{L}_{\mathrm{vel}}
+\lambda_{\mathrm{acc}}\mathcal{L}_{\mathrm{acc}}
+\lambda_{\mathrm{yaw}}\mathcal{L}_{\mathrm{yaw}}
+\lambda_{\circ}\mathcal{L}_{\circ},
\end{equation}
where $\lambda_{\mathrm{vel}},\lambda_{\mathrm{acc}},\lambda_{\mathrm{yaw}},\lambda_{\circ}\ge 0$ are loss weights.
Define the channel-wise error $e_{s,c}\triangleq \widehat{Y}_{s,c}-Y_{s,c}$ with
$\widehat{Y}_{s,c}\triangleq \widehat{\mathbf{Y}}_{t+1:t+H_{\mathrm{pred}}}[s,c]$ and
$Y_{s,c}\triangleq \mathbf{Y}_{t+1:t+H_{\mathrm{pred}}}[s,c]$.
Let $\Delta e_{s,c}\triangleq e_{s,c}-e_{s-1,c}$ ($s\ge2$) and
$\Delta^2 e_{s,c}\triangleq \Delta e_{s,c}-\Delta e_{s-1,c}$ ($s\ge3$). Then
\begin{subequations}\label{eq:smooth_yaw_losses}
\begin{align}
\mathcal{L}_{\mathrm{vel}}
&=\frac{1}{(H_{\mathrm{pred}}-1)C}\sum_{s=2}^{H_{\mathrm{pred}}}\sum_{c=1}^{C}
\HuberFunc{\Delta e_{s,c}}{\delta}, \label{eq:Lvel_short}\\
\mathcal{L}_{\mathrm{acc}}
&=\frac{1}{(H_{\mathrm{pred}}-2)C}\sum_{s=3}^{H_{\mathrm{pred}}}\sum_{c=1}^{C}
\HuberFunc{\Delta^2 e_{s,c}}{\delta}. \label{eq:Lacc_short}
\end{align}
\end{subequations}

For yaw, we compute $\psi_{t+s}=\atanTwo(Y_{\sin,t+s},Y_{\cos,t+s})$ and
$\widehat{\psi}_{t+s|t}=\atanTwo(\widehat{Y}_{\sin,t+s|t},\widehat{Y}_{\cos,t+s|t})$ and set
\begin{equation}\label{eq:Lyaw_short}
\mathcal{L}_{\mathrm{yaw}}
=\frac{1}{H_{\mathrm{pred}}}\sum_{s=1}^{H_{\mathrm{pred}}}
\HuberFunc{\mathrm{wrap}_{\pi}\!\big(\widehat{\psi}_{t+s|t}-\psi_{t+s}\big)}{\delta},
\end{equation}
where $\mathrm{wrap}_{\pi}(\cdot)$ maps angles to $(-\pi,\pi]$.
Finally, $\mathcal{L}_{\circ}$ enforces $(\widehat{Y}_{\sin})^2+(\widehat{Y}_{\cos})^2\approx 1$.
The LLM backbone is kept frozen, and only lightweight task adapters are trained.

The certificates in Lemma~\ref{lem:TI_certificate} and \eqref{eq:stat_cert} require an uncertainty description of the target-slot pointing mismatch $\Delta\boldsymbol\omega_{\tau}$ defined in \eqref{eq:delta_omega_def}.
For each {forecast time} $t$ and horizon $h$, the corresponding future slot is $\tau=t+h$.
We perform a one-time offline calibration on a separate validation set by collecting $\Delta\boldsymbol\omega_{\tau}$ over the target-slot horizons $h\in\{d{+}1,\ldots,H_{\mathrm{pred}}\}$, which are the horizons used to program the analog beamformer sequence over $\mathcal{T}_t$.
This offline calibration yields fixed uncertainty parameters that are reused during online operation and do not add online computation.
The use of a separate calibration set and empirical quantiles follows standard forecasting practice \cite{NEURIPS2019_5103c358} and is related to recent conformal methods for multi-step time-series forecasting \cite{SOUSA2024128434}.

\emph{Deterministic calibration.}
On an independent calibration split, we slide the same windowing protocol and, for each {forecast time} $t$, compute the target-window maximum
\begin{equation}
Z_t \triangleq
\max_{h\in\{d{+}1,\ldots,H_{\mathrm{pred}}\}}
\big\|\Delta\boldsymbol\omega_{t+h|t}\big\|_2,
\label{eq:calib_Zt}
\end{equation}

We then set
\begin{equation}
\delta_\omega \triangleq \operatorname{Quantile}_{1-\rho}\!\big(\{Z_t\}\big),
\label{eq:delta_omega_quantile}
\end{equation}
where $\{Z_t\}$ collects the target-window maxima computed over all {forecast time}s $t$ in the calibration split. 
{This calibrated radius directly interfaces with Lemma~\ref{lem:TI_certificate}, which demonstrates that the analog-beam main-lobe gain loss is uniformly bounded over all target-slot horizons by $L_k^2\delta_\omega^2$ for user $k$ under the small-detuning approximation and the target-window event. 
Therefore, choosing $\delta_\omega$ through the empirical $(1-\rho)$-quantile provides a direct way to translate observed target-window {forecast errors} into a high-confidence analog-pointing robustness margin. 
A larger confidence level $(1-\rho)$ yields a larger calibrated radius $\delta_\omega$ and leads to a more conservative but more robust analog steering certificate, an effect that is later quantified in the calibration-sensitivity experiments.
Together with the steering-angle event in Assumption A2, the resulting target-window guarantee holds with probability at least $1-(\rho+\rho_s)$, as detailed in the proof of Lemma~\ref{lem:TI_certificate} in the supplementary material.}

\emph{Statistical calibration:}
We estimate $(\boldsymbol\mu_\omega,\boldsymbol\Sigma_\omega)$ as the sample mean and covariance of $\{\Delta\boldsymbol\omega_{t+h}\}_{h=d+1}^{H_{\mathrm{pred}}}$ on the validation set and use them in \eqref{eq:stat_cert}. 
Both $\delta_\omega$ and $(\boldsymbol\mu_\omega,\boldsymbol\Sigma_\omega)$ are computed offline and then treated as fixed parameters during online operation.
In contrast, $\mathbf{J}_{\xi,k}$ in \eqref{eq:Jxi_closed} is evaluated online at the operating point associated with the analog beamformer applied in the corresponding target-slot.
{These offline estimates provide the uncertainty descriptors required by the deterministic and statistical certificates in the subsequent robust beamforming formulation without introducing additional online calibration cost.}

\section{Learning-Assisted QoS-driven Downlink Beamforming}
\label{sec:algorithm}

With proactive analog beam steering in Sec.~\ref{sec:VL-LLM}, the analog beamformer used in the target-slot $\tau$ is set as $\mathbf{A}_\tau\triangleq \mathbf{A}_{\tau|t^\star(\tau)}$ and treated as fixed during online operation.
The remaining online task is QoS-driven downlink digital beamforming with admission control under instantaneous QoS and transmit-power constraints.
This task is challenging because the digital beamforming and admission decisions must be computed online within each slot.
Recent model-based iterative solvers, such as the WMMSE and SCA, can be accurate but typically require multiple updates and repeated numerical solves \cite{10977774,11130524,11010920,ZHANG2026156209,ZHU2026103001}.
Learning-based {forecaster} are fast, yet they may output infeasible admissions and beamformers without explicit safeguards \cite{10333602,11148833}.
To bridge this gap, we develop a learning-assisted framework that produces per-slot digital beamforming and admission decisions with strict feasibility guarantees.
Robustness against {forecast-induced} pointing errors is handled offline through calibrated analog-pointing certificates, while the online stage focuses on satisfying instantaneous QoS and transmit-power constraints via a lightweight repair routine.

We adopt the bound-based pointing certificate and use it as a deterministic robustness certificate.
In each target-slot, $\mathbf{A}_\tau\in\mathcal{A}^{\mathrm{pred}}$ is produced by VL-LLM from the {HAP attitude forecasts}, and the corresponding pointing loss is controlled by the offline-calibrated parameters in Lemma~\ref{lem:TI_certificate}.
Accordingly, this section focuses on the remaining online variables $(\boldsymbol{\alpha},\mathbf{D})$ given $\mathbf{A}_\tau$.
{
Our goal is to avoid running a full iterative solver in every slot.
To maintain feasibility without high complexity, we employ a lightweight repair step.}

\subsection{Target-Slot QoS-driven Beamforming Formulation}
\label{subsec:routeA_overview}

Consider a snapshot at a target-slot $\tau$ and omit the slot index for brevity.
When serving all users is infeasible under $P_{\max}$, we introduce admission variables $\alpha_k\in\{0,1\}$.
For fixed $\mathbf{A}_\tau$, define the per-user QoS violation
\begin{equation}
v_k(\mathbf{D},\alpha_k) \triangleq \alpha_k r_k^{\min}-R_k(\mathbf{A}_\tau,\mathbf{D}), \qquad k\in\mathcal{K}.
\label{eq:qos_violation_def}
\end{equation}

The QAR objective is $\frac{1}{K}\sum_{k=1}^{K}\alpha_k$.
Feasibility additionally requires $\|\mathbf{A}_\tau\mathbf{D}\|_F^2\le P_{\max}$.
Once $\mathbf{A}_\tau$ is fixed, the pointing certificate is enforced through offline-calibrated constants and a simple pre-screening rule, and it does not introduce any additional online optimization variables.
For efficient evaluation, define $\mathbf{H}_{\mathrm{eff}}\triangleq \mathbf{H}^{\mathrm H}\mathbf{A}_\tau\in\mathbb{C}^{K\times N_{\mathrm{RF}}}$ and
\begin{equation}
\mathbf{G}\triangleq \mathbf{H}_{\mathrm{eff}}\mathbf{D}\in\mathbb{C}^{K\times K},\qquad
G_{k,j}=\mathbf{h}_k^{\mathrm H}\mathbf{A}_\tau\mathbf{d}_j,
\label{eq:G_def}
\end{equation}
so that $\mathrm{SINR}_k=\frac{|G_{k,k}|^2}{\sum_{j\neq k}|G_{k,j}|^2+\sigma^2}$ and $R_k=B\log_2(1+\mathrm{SINR}_k)$.
Let $\mathbf{h}_{\mathrm{eff},k}\triangleq (\mathbf{H}_{\mathrm{eff}})_{k,:}^{\mathrm H}\in\mathbb{C}^{N_{\mathrm{RF}}}$ denote the effective channel vector for user $k$.

\subsection{Solver-Guided Labels and KKT Targets}
\label{subsec:teacher_impl_aligned}

We design an offline {reference solver} to generate supervision by computing feasible high-quality solutions to \eqref{prob:QAR_static} and \eqref{prob:EE_given_adm} using a lightweight iterative solver. 
The {offline reference solver} follows a two-stage routine and stores the resulting admission decisions, beamformers, and auxiliary variables as distillation labels. 
Since $L_k$ and $\delta_\omega$ are fixed constants calibrated offline, the pointing certificate can be enforced by a simple pre-screening rule. 
Define the certified user set
\begin{equation}
\mathcal{K}_{\mathrm{cert}} \triangleq \left\{k\in\mathcal{K}: L_k^2\,\delta_\omega^2 \le \epsilon \right\}.
\end{equation}

Only users in $\mathcal{K}_{\mathrm{cert}}$ are eligible for admission in the target-slot optimization, and we set $\alpha_k=0$ for all $k\notin\mathcal{K}_{\mathrm{cert}}$.
Accordingly, the admission and strict-repair procedures operate on $\mathcal{K}_{\mathrm{cert}}$ without introducing additional online variables.

\textbf{Stage I. Feasibility-driven admission.}
Starting from the certified candidate set $\mathcal{K}_{\mathrm{cand}}=\mathcal{K}_{\mathrm{cert}}$, the {offline reference solver} iteratively removes users until the snapshot becomes feasible under $P_{\max}$.
To decide which user to be removed, the {offline reference solver} uses a low-cost required-power proxy $\pi_k$ computed from the snapshot.
For reproducibility and low online complexity, we define $\pi_k$ as
\begin{equation}
\gamma_k \triangleq 2^{r_k^{\min}/B}-1,\qquad
\pi_k \triangleq \frac{\gamma_k \sigma^2}{\left|\mathbf{h}_k^{\mathrm H}\mathbf{A}_\tau\mathbf{f}_k\right|^2+\varepsilon_{\pi}},
\label{eq:pi_def}
\end{equation}
where $\mathbf{f}_k$ is a low-complexity reference beam, e.g., the matched-filter $\mathbf{f}_k=\mathbf{h}_{\mathrm{eff},k}/\|\mathbf{h}_{\mathrm{eff},k}\|_2$, and $\varepsilon_{\pi}>0$ is a small constant.

This proxy approximates the transmit power required to meet the SINR target under a noise-limited approximation and is used only for ranking and removal.
{The admission priority is determined by the snapshot-wise required-power proxy $\pi_k$, which jointly reflects effective channel quality, QoS target, and current analog steering state.
Users with larger $\pi_k$ are more costly to satisfy and are removed first in Stage I.}
In each iteration, the {offline reference solver} removes the user with the largest $\pi_k$ from $\mathcal{K}_{\mathrm{cand}}$ and repeats until feasibility is reached or $\mathcal{K}_{\mathrm{cand}}=\emptyset$, yielding an admission label $\boldsymbol{\alpha}^\star$.
We denote the resulting admitted set by $\mathcal{K}_{\mathrm{adm}}^\star=\{k:\alpha_k^\star=1\}$.

\textbf{Stage II. Utility refinement over the admitted set.}
Given the admitted set $\mathcal{K}_{\mathrm{adm}}^\star$, the {offline reference solver} refines the Stage-II utility or an equivalent sum-rate surrogate using a small number of iterations of a local solver.
A typical choice is the WMMSE method or SCA on $\mathbf{H}_{\mathrm{eff}}$ \cite{10977774,11130524,11010920,ZHANG2026156209,ZHU2026103001}.
The refinement maintains $R_k\ge r_k^{\min}$ for all $k\in\mathcal{K}_{\mathrm{adm}}^\star$ and $\|\mathbf{A}_\tau\mathbf{D}\|_F^2\le P_{\max}$.
The {offline reference solver} outputs a feasible beamformer $\mathbf{D}^\star$ and auxiliary variables that are consistent with the local KKT conditions of the refinement solver.
For WMMSE-type updates, we store per-user scalars $u_k^\star$ and $w_k^\star$ and the power-dual scalar $\nu^\star$, and denote the collected auxiliary labels by $\boldsymbol{\zeta}^\star$.
These auxiliary variables are aligned with the stationarity structure of WMMSE-type KKT conditions and are used to supervise the subsequent closed-form beamformer reconstruction.

For each snapshot, the {offline reference solver} provides $\boldsymbol{\alpha}^\star$, $\mathbf{D}^\star$, the proxy vector $\boldsymbol{\pi}$ computed by \eqref{eq:pi_def},
and KKT-related labels $\boldsymbol{\zeta}^\star$, such as $\{u_k^\star,w_k^\star\}$ and $\nu^\star$.

\subsection{KKT-Guided Reconstruction With Strict Feasibility Repair}
\label{subsec:student_kkt_recon}

Next, we design a compact {online forecaster} that outputs soft admission scores $\widehat{\boldsymbol{\alpha}}\in[0,1]^K$ and auxiliary variables $\widehat{\boldsymbol{\zeta}}$ from snapshot information. 
The inputs include the instantaneous effective CSI $\mathbf{H}_{\mathrm{eff}}$, the QoS targets $\{r_k^{\min}\}$, and pointing-uncertainty proxies such as $\widehat{\sigma}^2_{\xi,k}$ in \eqref{eq:sigma_xi_scalar}. 
Admission decisions are binarized as $\alpha_k=\mathbb{I}\{\widehat{\alpha}_k\ge \eta\}$ using a fixed threshold $\eta\in(0,1)$, where $\mathbb{I}\{\cdot\}$ denotes the indicator function.

Instead of directly regressing the digital beamformer $\mathbf{D}$, we reconstruct it using a WMMSE-inspired closed-form mapping that is consistent with the corresponding KKT conditions. 
Specifically, let $\widehat{\boldsymbol{\zeta}}$ include per-user scalars $(\widehat{u}_k,\widehat{w}_k)$ and an initial estimate of the power-dual variable $\widehat{\nu}\ge 0$.
Define
\begin{align}
\mathbf{C}(\nu)
\triangleq
\sum_{k=1}^{K} \alpha_k\,\widehat{w}_k\,|\widehat{u}_k|^2 \,
\mathbf{h}_{\mathrm{eff},k}\mathbf{h}_{\mathrm{eff},k}^{\mathrm H}
+\nu \mathbf{I},\notag\\
\qquad
\mathbf{d}_k(\nu)
\triangleq
\mathbf{C}(\nu)^{-1}\Big(\alpha_k\,\widehat{w}_k\,\widehat{u}_k^{\!*}\,\mathbf{h}_{\mathrm{eff},k}\Big),
\label{eq:kkt_recon_wmmse}
\end{align}
where $\mathbf{D}(\nu)=[\mathbf{d}_1(\nu),\ldots,\mathbf{d}_K(\nu)]$.
We select $\nu$ by a few bisection steps, since $\|\mathbf{A}_\tau\mathbf{D}(\nu)\|_F^2$ is non-increasing in $\nu$ for this WMMSE-type mapping.

To guarantee the instantaneous power constraint, we further apply the scaling projection
\begin{equation}
\mathbf{D}\leftarrow \mathbf{D}\cdot \min\!\left\{1,\sqrt{\frac{P_{\max}}{\|\mathbf{A}_\tau\mathbf{D}\|_F^2+\varepsilon_p}}\right\},
\label{eq:routeA_power_proj}
\end{equation}
which enforces $\|\mathbf{A}_\tau\mathbf{D}\|_F^2\le P_{\max}$ by construction and $\varepsilon_p>0$ is a small constant.

Even with the above reconstruction, QoS constraints may still be violated due to model mismatch, finite-precision effects, or an overly aggressive admission decision.
We therefore apply a deterministic strict-repair routine based on worst-first user removal followed by an add-back step.
The routine terminates in at most $K$ removals because the admitted set strictly decreases and the empty set is always feasible.
{Moreover, the power projection in \eqref{eq:routeA_power_proj} enforces the instantaneous transmit-power constraint by design. 
The removal loop iteratively reduces the admitted set and terminates once a feasible solution is found. 
Consequently, the repaired output satisfies both the instantaneous power constraint and the QoS requirements of all retained users.}

Let $\mathcal{K}_{\mathrm{adm}}=\{k:\alpha_k=1\}$ and compute the QoS gap
$g_k=\big(r_k^{\min}-R_k(\mathbf{A}_\tau,\mathbf{D})\big)_+$ for $k\in\mathcal{K}_{\mathrm{adm}}$.
If all $g_k=0$, the current solution is feasible.
Otherwise, we remove the user with the largest normalized violation,
\begin{equation}
k^\dagger = \arg\max_{k\in\mathcal{K}_{\mathrm{adm}}} \ \frac{g_k}{\pi_k+\varepsilon_{\pi}},
\label{eq:drop_metric}
\end{equation}
where $\pi_k$ is the required-power proxy in \eqref{eq:pi_def}.

We set $\alpha_{k^\dagger}\leftarrow 0$, update $\mathcal{K}_{\mathrm{adm}}\leftarrow \mathcal{K}_{\mathrm{adm}}\setminus\{k^\dagger\}$, and reconstruct $\mathbf{D}$ again.
We repeat until feasibility is reached or $\mathcal{K}_{\mathrm{adm}}=\emptyset$.
After feasibility is achieved, we add back removed users in ascending $\pi_k$ order and keep a user only if all constraints remain satisfied.
Finally, we run a few iterations of a local refinement initialized from the repaired $\mathbf{D}$, and accept an iteration only if it preserves $R_k\ge r_k^{\min}$ for all admitted users and satisfies \eqref{eq:routeA_power_proj}.

\subsection{Distillation-Based Training Objective}
\label{subsec:trainloss_aligned}

The {online forecaster} is trained to distill the admission decisions {generated by the offline reference solver} and KKT-related auxiliary variables, while penalizing QoS violations.
For admission distillation, we use the binary cross-entropy loss
$\mathcal{L}_{\mathrm{adm}}=-\frac{1}{K}\sum_{k=1}^{K}\!\left[\alpha_k^\star\log(\sigma(\widehat{\alpha}_k))+(1-\alpha_k^\star)\log\!\big(1-\sigma(\widehat{\alpha}_k)\big)\right]$,
where $\sigma(\cdot)$ is the sigmoid.
For the utility-related term $\mathcal{L}_{\mathrm{util}}$, a typical choice is the regression loss $\|\widetilde{\mathbf{D}}-\mathbf{D}^\star\|_F^2$, or a negative-utility surrogate evaluated on $(\widetilde{\boldsymbol{\alpha}},\widetilde{\mathbf{D}})$.
The overall beamforming distillation objective is defined as
\begin{align}
\mathcal{L}_{\mathrm{bf}}
&=
\lambda_{\mathrm{adm}}\,\mathcal{L}_{\mathrm{adm}}(\widehat{\boldsymbol{\alpha}},\boldsymbol{\alpha}^\star)
+\lambda_{\zeta}\,\|\widehat{\boldsymbol{\zeta}}-\boldsymbol{\zeta}^{\star}\|_2^2\notag\\
&\quad+\lambda_{\mathrm{qos}}\sum_{k=1}^{K}\varpi_k\,
\phi_{\beta}\!\Big(\widetilde{\alpha}_k r_k^{\min}-R_k(\mathbf{A}_\tau,\widetilde{\mathbf{D}})\Big)
\notag\\
&\quad
-\lambda_{\mathrm{qar}}\cdot \frac{1}{K}\sum_{k=1}^{K}\widetilde{\alpha}_k
+\lambda_{\mathrm{util}}\;\mathcal{L}_{\mathrm{util}}(\widetilde{\mathbf{D}},\mathbf{D}^\star),\notag\\
&\quad
\varpi_k=1+c_{\omega}\widehat{\sigma}^2_{\xi,k}.
\label{eq:trainloss_aligned}
\end{align}
where $\phi_{\beta}(x)=\frac{1}{\beta}\log(1+\exp(\beta x))$ is a smooth penalty for QoS violations.
The pair $(\widetilde{\boldsymbol{\alpha}},\widetilde{\mathbf{D}})$ denotes the final outputs after KKT-guided reconstruction and the power projection in \eqref{eq:routeA_power_proj}.
The vector $\widetilde{\boldsymbol{\alpha}}$ is the binary admission decision after thresholding and strict feasibility repair.
The uncertainty-aware weight $\varpi_k=1+c_{\omega}\widehat{\sigma}^2_{\xi,k}$ assigns a larger QoS-violation cost to users with higher pointing uncertainty.

We first warm up the {online forecaster} using admission and auxiliary distillation, then enable the feasibility-aware term.
The strict repair is always applied at inference.
It can also be enabled during training-time evaluation to reduce train-test mismatch.
Overall, the main steps of the proposed target-slot QoS-driven beamforming method are summarized in Algorithm~\ref{alg:routeA_aligned}.

\subsection{Algorithm and Complexity}

\begin{algorithm}[!t]
\caption{QoS-driven downlink beamforming with KKT-guided reconstruction and strict repair}
\label{alg:routeA_aligned}
\footnotesize
\begin{algorithmic}[1]
\Require Snapshot $(\mathbf{H}_\tau,\mathbf{A}_\tau)$, QoS targets $\{r_k^{\min}\}$, $P_{\max}$
\State $\mathbf{H}_{\mathrm{eff}}\leftarrow \mathbf{H}_\tau^{\mathrm H}\mathbf{A}_\tau$
\State Extract snapshot features from CSI $\mathbf{H}_{\mathrm{eff}}$, QoS targets, and uncertainty proxies
\State {The online forecaster} outputs $(\widehat{\boldsymbol{\alpha}},\widehat{\boldsymbol{\zeta}})$ and binarize $\alpha_k=\mathbb{I}\{\widehat{\alpha}_k\ge\eta\}$
\State Compute $\pi_k$ via \eqref{eq:pi_def} and set $\mathcal{K}_{\mathrm{adm}}\leftarrow \{k:\alpha_k=1\}$
\Repeat
  \State Reconstruct $\mathbf{D}$ via \eqref{eq:kkt_recon_wmmse}
  \State Choose $\nu$ by bisection and apply \eqref{eq:routeA_power_proj}
  \State Compute $g_k=\big(r_k^{\min}-R_k(\mathbf{A}_\tau,\mathbf{D})\big)_+$ for $k\in\mathcal{K}_{\mathrm{adm}}$
  \If{all $g_k=0$} \State \textbf{break} \EndIf
  \State Set $k^\dagger$ by \eqref{eq:drop_metric}, update $\alpha_{k^\dagger}\leftarrow 0$, and update $\mathcal{K}_{\mathrm{adm}}\leftarrow \mathcal{K}_{\mathrm{adm}}\setminus\{k^\dagger\}$
\Until{$\mathcal{K}_{\mathrm{adm}}=\emptyset$}
\State Attempt to add back removed users in ascending $\pi_k$ order and keep a user if feasibility is preserved
\State QoS-safe refinement iterations
\State Output repaired $(\widetilde{\boldsymbol{\alpha}},\widetilde{\mathbf{D}})$
\end{algorithmic}
\end{algorithm}

We report the per-snapshot computational complexity of the proposed online procedure.
Let $K$ be the number of users, $N_{\mathrm{RF}}$ the number of RF chains, and $M$ the number of antennas.
Let $N_{\mathrm{bis}}$ be the number of bisection steps for selecting the power-dual variable $\nu$.
Let $N_{\mathrm{drop}}$ be the number of worst-first removal rounds in the strict-repair loop, where $0\le N_{\mathrm{drop}}\le K$.
Let $N_{\mathrm{ref}}$ be the number of capped refinement iterations.

Forming $\mathbf{H}_{\mathrm{eff}}=\mathbf{H}_\tau^{\mathrm H}\mathbf{A}_\tau$ costs $O(KMN_{\mathrm{RF}})$ in general, and this step can be skipped if $\mathbf{H}_{\mathrm{eff}}$ is directly available.
For a fixed $\nu$, constructing $\mathbf{C}(\nu)$ costs $O(KN_{\mathrm{RF}}^2)$.
Solving the linear system for all users can be implemented by one matrix factorization with cost $O(N_{\mathrm{RF}}^3)$ followed by $K$ back-substitutions with cost $O(KN_{\mathrm{RF}}^2)$.
Therefore, the reconstruction cost per $\nu$ is $O(N_{\mathrm{RF}}^3+KN_{\mathrm{RF}}^2)$.
Evaluating $\mathbf{G}=\mathbf{H}_{\mathrm{eff}}\mathbf{D}$ and the resulting SINR/rates costs $O(K^2N_{\mathrm{RF}})$.
The {online forecaster} network inference can be written as $O(C_{\mathrm{NN}})$ and is typically dominated by the above matrix operations.

Overall, the bisection contributes $O\!\big(N_{\mathrm{bis}}(N_{\mathrm{RF}}^3+KN_{\mathrm{RF}}^2)\big)$ per reconstruction.
The strict repair repeats one reconstruction and one rate evaluation for at most $N_{\mathrm{drop}}$ rounds.
The refinement stage runs for $N_{\mathrm{ref}}$ iterations with the same order of matrix operations and rate evaluation.
As a result, the total per-snapshot online complexity is
$O\!\Big(KMN_{\mathrm{RF}}+(N_{\mathrm{drop}}+1)\big(N_{\mathrm{bis}}(N_{\mathrm{RF}}^3+KN_{\mathrm{RF}}^2)+K^2N_{\mathrm{RF}}\big)+N_{\mathrm{ref}}(N_{\mathrm{RF}}^3+KN_{\mathrm{RF}}^2+K^2N_{\mathrm{RF}})+C_{\mathrm{NN}}\Big)$.

In hybrid beamforming, $N_{\mathrm{RF}}$ is typically small and we cap $(N_{\mathrm{bis}},N_{\mathrm{drop}},N_{\mathrm{ref}})$ by small constants.
This keeps the online computation bounded per slot.
{Therefore, the proposed online stage avoids expensive iterative optimization, allowing admission and digital beamforming decisions to be computed within strict per-slot latency constraints.}
The {offline reference solver} cost is incurred only during training and is amortized through distillation.

\section{Simulation Results}\label{sec:results}

We present two complementary sets of simulations.
First, we evaluate VL-LLM on multi-horizon attitude forecasting using HAP flight telemetry and quantify its {effect on proactive} analog beamforming over the target-slot window.
Second, we benchmark the proposed QoS-driven beamforming solver on synthetic multiuser snapshots under strict instantaneous QoS and transmit-power constraints.

\subsection{Comparison Algorithms and Simulation Parameters}
\label{sec:sim_setup}

For attitude {forecasting}, we compare the proposed VL-LLM with TimeLLM~\cite{jin2024timellm}, PatchTST~\cite{Yuqietal-2023-PatchTST}, TimesNet~\cite{wu2023timesnet}.
{We also evaluate prompt ablations, input modes. 
Robustness is further tested against inertial measurement unit (IMU) noise, calibration confidence, and inference latency.
Performance is measured by mean absolute error (MAE) and root mean square error (RMSE) in degrees, along with the horizon-$H_{\mathrm{pred}}$ mean MAE in the target-slot window.
For calibration, we use the scalar mismatch magnitude $\|\Delta\boldsymbol{\omega}\|_2$ in \eqref{eq:calib_Zt}.}

{For QoS-driven beamforming, we compare the proposed learning-assisted solver with proximal policy optimization (PPO), PPO-$\lambda$~\cite{10681823,ADAM2025102771}, and a denoising diffusion probabilistic model (DDPM)-based generative baseline~\cite{10750041,11169307}.
To evaluate attitude compensation impact, we consider no compensation, reactive, forecast-based, and ideal modes. 
We also test under different channel conditions, user distributions, and admission priorities.}

{Table~\ref{tab:sim_settings} summarizes the main experimental settings.
These include the real-flight telemetry source, the delay-aware {forecasting} setup, the backbone and training configuration, and the communication evaluation settings.}

\begin{table*}[t]
\centering
\footnotesize
\caption{{Main experimental settings for {forecasting} and communication evaluation.}}
\label{tab:sim_settings}
\setlength{\tabcolsep}{5pt}
\begin{tabular}{p{0.24\linewidth} p{0.68\linewidth}}
\toprule
{\textbf{Category}} & {\textbf{Setting}} \\
\midrule
{{Forecasting} data \& timing} &
{Real-flight telemetry, $f_s = 10~Hz$, train/validation/test split = $70\%/10\%/20\%$, $L_{\mathrm{win}}=192$, $L_{\mathrm{label}}=24$, $H_{\mathrm{pred}}=12$, decision delay $d=6$, target-slot window $h\in\{7,\ldots,12\}$} \\
{Forecasting baselines} &
{TimeLLM, PatchTST, TimesNet} \\
{Ablations and robustness} &
{task-only, task+physics, visual-only, numeric-only; IMU-noise, calibration-confidence, and latency tests} \\
{Forecasting model \& training} &
{DeepSeek-R1-7B, ViLT, AdamW, early stopping patience = 15, BF16 mixed precision} \\
{Loss weights} & 
{$(\lambda_{\mathrm{adm}},\lambda_{\zeta},\lambda_{\mathrm{qos}},\lambda_{\mathrm{qar}},\lambda_{\mathrm{util}},c_{\omega}) = (2, 1, 20, 0.1, 2, 0)$, $(\lambda_{\mathrm{vel}},\lambda_{\mathrm{acc}},\lambda_{\mathrm{yaw}},\lambda_{\circ}) = (0.02, 0, 0, 0)$} \\
{Communication baselines} &
{DualFormer (proposed), PPO, PPO-$\lambda$, DDPM} \\
{Evaluation modes} &
{none, reactive, forecast-based, ideal, Rician strong, Rician weak; uniform, clustered, edge-biased, forecasted-QoS-difficulty, channel-gain, random; $K=10$} \\
{Digital beamforming solver setting} &
{$K=10$ users; 12$\times$12 UPA; validation: seed 2025, 2048 samples; test: seed 2026, 2048 samples} \\
{Common post-processing} &
{Admission gate with $k_{\min}=8$, deterministic strict repair, QoS-safe WMMSE refinement with 10 iterations, DDPM uses best-of-32 sampling and two WMMSE warm-start iterations} \\
{Hardware} &
{NVIDIA RTX 5090 GPU, 32 GB memory} \\
\bottomrule
\end{tabular}
\end{table*}

\subsection{Performance Evaluation}
\label{sec:performance_eval}

\subsubsection{VL-LLM for Delay-Aware Attitude Forecasting}

{We first evaluate whether the proposed VL-LLM can provide sufficiently accurate short-term attitude forecasts for delay-aware proactive beam steering. 
We compute MAE and RMSE in degrees and evaluate yaw errors using wrap-to-$(-180^\circ,180^\circ]$. 
We focus on the target-slot windoww $h\in\{d+1,\ldots,H_{\mathrm{pred}}\}$. 
We report horizon-$H_{\mathrm{pred}}$ metrics to reflect forecast-based beam-steering quality under the modeled decision delay.}

\begin{figure*}[t]
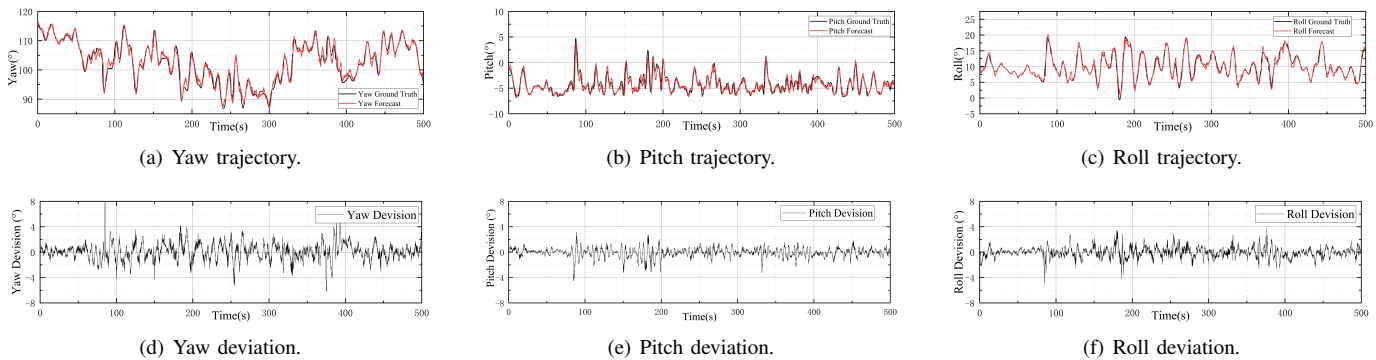

  \centering
  \subfigure[Yaw trajectory.]{%
    \includegraphics[width=0.31\linewidth]{Yaw.png}%
    \label{fig:yaw_pred}%
  }\hfill
  \subfigure[Pitch trajectory.]{%
    \includegraphics[width=0.31\linewidth]{Pitch.png}%
    \label{fig:pitch_pred}%
  }\hfill
  \subfigure[Roll trajectory.]{%
    \includegraphics[width=0.31\linewidth]{Roll.png}%
    \label{fig:roll_pred}%
  }\\
  \subfigure[Yaw deviation.]{%
    \includegraphics[width=0.31\linewidth]{Yaw_Devision.png}%
    \label{fig:yaw_dev}%
  }\hfill
  \subfigure[Pitch deviation.]{%
    \includegraphics[width=0.31\linewidth]{Pitch_Devision.png}%
    \label{fig:pitch_dev}%
  }\hfill
  \subfigure[Roll deviation.]{%
    \includegraphics[width=0.31\linewidth]{Roll_Devision.png}%
    \label{fig:roll_dev}%
  }
  \caption{{Attitude trajectories and forecasting deviations for yaw, pitch, and roll.}}
  \label{fig:pred_traj_dev}
\end{figure*}

\begin{table}[t]
\centering
\footnotesize
\caption{{Forecasting performance, ablation, robustness, and latency analysis.}}
\label{tab:pred_all_compact}
\begin{tabular}{lcc}
\toprule
Setting & MAE & RMSE \\
\midrule
\multicolumn{3}{c}{Main baseline comparison} \\
\midrule
VL-LLM (Proposed) & \textbf{0.3722} & \textbf{0.6188} \\
TimeLLM~\cite{jin2024timellm} & 0.5880 & 0.8990 \\
PatchTST~\cite{Yuqietal-2023-PatchTST} & 0.5890 & 0.8770 \\
TimesNet~\cite{wu2023timesnet} & 0.9830 & 1.9430 \\
\midrule
\multicolumn{3}{c}{{Prompt ablation}} \\
\midrule
{Task-only} & {0.3872} & {0.6295} \\
{Task+physics} & {0.3851} & {0.6328} \\
\midrule
\multicolumn{3}{c}{{Input-mode comparison}} \\
\midrule
{Visual-only} & {0.4735} & {0.7673} \\
{Numeric-only} & {0.4299} & {0.7072} \\
\midrule
\multicolumn{3}{c}{{IMU-noise robustness}} \\
\midrule
{Noise std = 0.005} & {0.3799} & {0.6305} \\
{Noise std = 0.010} & {0.3799} & {0.6302} \\
{Noise std = 0.020} & {0.3800} & {0.6306} \\
\midrule
\multicolumn{3}{c}{{Latency profiling}} \\
\midrule
{Total mean (ms)} & \multicolumn{2}{c}{{36.24}} \\
{Total p99 (ms)} & \multicolumn{2}{c}{{40.13}} \\
\bottomrule
\end{tabular}
\end{table}

Fig.~\ref{fig:pred_traj_dev} shows that the {forecast} yaw, pitch, and roll trajectories closely follow the ground truth over representative test segments.
{Analog pointing is compensated over the actionable target-slot window defined in Sec. II-A1. 
Therefore, horizon-aligned errors are more relevant than one-step accuracy for assessing proactive beam steering quality.}

Figs.~\ref{fig:yaw_pred}, \ref{fig:pitch_pred}, and \ref{fig:roll_pred} show that VL-LLM closely follows the ground truth for yaw, pitch, and roll over representative test trajectories. 
The corresponding deviation plots in Figs.~\ref{fig:yaw_dev}, \ref{fig:pitch_dev}, and \ref{fig:roll_dev} show that the errors remain concentrated around zero and stay controlled at the forecast horizon. 
Specifically, at $h=12$, the absolute errors are within $\pm4^\circ$ for {$97.10\%$, $98.93\%$, and $99.52\%$ of yaw, pitch, and roll samples, respectively.
The corresponding 95th-percentile absolute errors at $3.29^\circ$, $1.99^\circ$, and $1.98^\circ$.} 
These horizon-aligned tail errors directly support the offline calibration of $\delta_\omega$ in \eqref{eq:delta_omega_quantile}, thereby fixing the uncertainty parameters used for robust beamforming.
{To quantify calibration-confidence sensitivity, we vary the confidence level \((1-\rho)\) from 0.80 to 0.99. 
The interval coverage increases from 0.9877 to 0.9968. 
Meanwhile, the MAE and RMSE remain unchanged. 
This confirms that higher calibration confidence mainly renders the offline robustness margin more conservative, without affecting the point-forecast accuracy of VL-LLM.}

{Table~\ref{tab:pred_all_compact} shows that VL-LLM achieves the best aggregate {forecasting} accuracy among the compared baselines.
Compared with the strongest learning baseline PatchTST, VL-LLM reduces the RMSE from $0.8770^\circ$ to $0.6188^\circ$, corresponding to a $29.4\%$ improvement.
It also reduces the MAE from $0.5890^\circ$ to $0.3722^\circ$.
These results confirm that the proposed {VL-LLM} provides a substantially more accurate basis for delay-aware proactive beam steering.
Removing prompt components or replacing the multimodal input with single-modality input degrades {forecasting} accuracy. 
This confirms that both structured task priors and multimodal telemetry representations are beneficial.}
{VL-LLM also remains stable under moderate IMU perturbations. 
From the reference model to the noisiest case, the RMSE increases from $0.6188^\circ$ to $0.6306^\circ$, an increase of less than $2\%$.}

\subsubsection{Communication Performance After Attitude Compensation}

{We then evaluate whether improved attitude forecasting translates into actual communication gains after compensation and beamforming refinement.
To this end, we compare different compensation modes.
We also test the method under different channel conditions, user distributions, and admission priorities.}

\begin{table}[t]
\centering
\footnotesize
\caption{{Communication performance under different compensation and system settings.}}
\label{tab:comm_closed_loop_compact}
\begin{tabular}{lcc}
\toprule
{Setting} & {QAR} & {Sum-rate} \\
\midrule
\multicolumn{3}{c}{{Compensation mode}} \\
\midrule
{No compensation} & {0.1984} & {25.6294} \\
{Reactive} & {0.1986} & {25.7547} \\
{forecast-based} & {0.1994} & {25.8342} \\
{Ideal} & {0.2000} & {28.0534} \\
\midrule
\multicolumn{3}{c}{{Channel condition}} \\
\midrule
{Rician weak} & {0.1910} & {24.1283} \\
{Rician strong} & {0.1994} & {25.8342} \\
\midrule
\multicolumn{3}{c}{{User distribution}} \\
\midrule
{Uniform} & {0.1994} & {25.8342} \\
{Clustered} & {0.1980} & {25.5361} \\
{Edge-biased} & {0.1986} & {25.7268} \\
\midrule
\multicolumn{3}{c}{{Admission priority}} \\
\midrule
{forecasted-QoS-difficulty} & {0.1994} & {25.8342} \\
{Channel-gain} & {0.1978} & {25.6584} \\
{Random} & {0.1990} & {25.7638} \\
\bottomrule
\end{tabular}
\end{table}

{Table~\ref{tab:comm_closed_loop_compact} summarizes the end-to-end communication results after attitude compensation under different compensation modes and system settings.
Forecast-based compensation consistently outperforms both no compensation and reactive modes, while the ideal mode remains the upper bound.
Compared with the reactive baseline, forecast-based compensation increases the QAR from $0.1986$ to $0.1994$ and the sum-rate from $25.7547$ to $25.8342$ under the same settings.
Although the gain is moderate, it is consistent and confirms that the proposed {VL-LLM} provides communication benefits.}

{The same table further shows that the proposed method remains effective under non-LoS (NLoS) uncertainty and different system settings.
The two Rician settings lead to different communication performance, with the strong-Rician case achieving higher QAR and sum-rate.
The clustered and edge-biased user layouts are slightly more challenging than the uniform layout. 
This reflects less balanced QoS difficulty under the same power budget.
Among the admission strategies, the forecasted-QoS-difficulty priority achieves the best overall tradeoff in both QAR and sum-rate.
This is consistent with its design objective of ranking users according to their instantaneous QoS difficulty under the current analog pointing and power budget.}

\subsubsection{Learning-Assisted QoS-driven Downlink Beamforming Solver}

{Finally, we benchmark the Stage-II digital beamforming solver itself using static multiuser LoS scenarios.
This experiment uses identical settings for user admission, deterministic repair, and power budget.}

\begin{table}[t]
\centering
\footnotesize
\caption{Benchmark of the QoS-driven digital beamforming solver.}
\label{tab:solver_test_metrics}
\begin{tabular}{lccc}
\toprule
Method & Feasible & QAR & Sum-rate \\
\midrule
DualFormer (Proposed) & \textbf{1.0000} & \textbf{0.9944} & \textbf{77.3820} \\
PPO & 0.9980 & 0.8100 & 68.5984 \\
PPO-$\lambda$ & 1.0000 & 0.6810 & 68.8129 \\
DDPM & 1.0000 & 0.8142 & 66.2873 \\
\bottomrule
\end{tabular}
\end{table}

{Table~\ref{tab:solver_test_metrics} shows that DualFormer achieves the best sum-rate and the highest QAR while maintaining near-unity feasibility under the same user admission, repair, and power budget.
Compared with the best-QAR baseline DDPM, DualFormer improves the QAR from $0.8142$ to $0.9944$.
Compared with the best-sum-rate baseline PPO-$\lambda$, DualFormer improves the sum-rate from $68.8129$ to $77.3820$, a gain of $12.5\%$.
These results support the design in Sec.~\ref{sec:algorithm}, where KKT-guided reconstruction reduces approximation error and deterministic repair prevents QoS violations.}

\section{Conclusion}
This paper investigated HAP downlink mmWave communications.
{It addressed platform attitude-induced beam misalignment and stringent per-slot online decision requirements}.
We developed a multimodal LLM-enabled beamforming framework for robust HAP downlink communications. 
Specifically, we designed a VL-LLM that learned from multivariate real-flight telemetry to {forecast short-term HAP attitudes under platform shaking}.
We also introduced an offline {forecast}-error calibration {procedure}.
This yielded reliable upper bounds on the residual errors.
{Based on the compensated analog beamformer, we further developed a learning-assisted digital beamforming and admission scheme.
It incorporated lightweight feasibility enforcement to satisfy instantaneous transmit-power and QoS constraints.
Using real-flight telemetry sampled at 10 Hz, the proposed framework achieved accurate short-term attitude forecasting with p99 latency of 40.13 ms. 
This enabled delay-aware proactive beam steering and consistent communication gains over reactive and no compensation baselines across diverse settings.
Future work will consider stronger NLoS effects, nonstationary sensing uncertainty, cooperative aerial-network settings.}

\appendices
\section{Proof of Lemma~\ref{lem:TI_certificate}}

\begin{proof}
Fix a user $k$ and consider the target-window horizons $h\in\{d{+}1,\ldots,H_{\mathrm{pred}}\}$ programmed at time $t$.
Under Assumption~A2, the forecasted steering angles satisfy $(\vartheta_k,\varphi_k)\in\mathcal{S}_k$ on an event with probability at least $1-\rho_s$.
Under the small-detuning regime where the quadratic approximation in Eq.~(\ref{eq:A2}) of the main manuscript holds, the normalized main-lobe gain loss can be expressed as
\begin{equation}
\Delta G_{\mathrm{UPA},k}
\approx
\Delta\boldsymbol\omega_{t+h|t}^{\top}\mathbf{Q}_k(\vartheta_k,\varphi_k)\Delta\boldsymbol\omega_{t+h|t}.
\end{equation}

Since $c_x,c_y>0$ and $\mathbf{Q}_k(\vartheta,\varphi)=\mathbf{J}_{\xi,k}(\vartheta,\varphi)^{\top}\mathrm{diag}(c_x,c_y)\mathbf{J}_{\xi,k}(\vartheta,\varphi)$, we have $\mathbf{Q}_k(\vartheta,\varphi)\succeq \mathbf{0}$.
Therefore, the Rayleigh quotient bound gives
\begin{equation}
\Delta\boldsymbol\omega_{t+h|t}^{\top}\mathbf{Q}_k(\vartheta_k,\varphi_k)\Delta\boldsymbol\omega_{t+h|t}
\le
\lambda_{\max}\!\big(\mathbf{Q}_k(\vartheta_k,\varphi_k)\big)\,
\big\|\Delta\boldsymbol\omega_{t+h|t}\big\|_2^2.
\label{eq:rayleigh_bound_appendix}
\end{equation}

By the definition of $L_k^2$ in Eq.~(24) of the main manuscript, for all $(\vartheta,\varphi)\in\mathcal{S}_k$, 
$\lambda_{\max}\!\big(\mathbf{Q}_k(\vartheta,\varphi)\big)\le L_k^2$.
Hence, whenever $(\vartheta_k,\varphi_k)\in\mathcal{S}_k$,
\begin{equation}
\Delta G_{\mathrm{UPA},k}
\le
L_k^2\,\big\|\Delta\boldsymbol\omega_{t+h|t}\big\|_2^2.
\label{eq:bound_Lk_appendix}
\end{equation}

Next, under Assumption~A1, the target-window event
\begin{equation}
\max_{h\in\{d{+}1,\ldots,H_{\mathrm{pred}}\}}
\big\|\Delta\boldsymbol\omega_{t+h|t}\big\|_2
\le
\delta_\omega
\end{equation}
holds with probability at least $1-\rho$.

On this event, \eqref{eq:bound_Lk_appendix} implies that for every $h\in\{d{+}1,\ldots,H_{\mathrm{pred}}\}$,
\begin{equation}
\Delta G_{\mathrm{UPA},k}
\le
L_k^2\,\delta_\omega^2.
\end{equation}

If $L_k^2\delta_\omega^2\le\epsilon$, then $\Delta G_{\mathrm{UPA},k}\le\epsilon$ holds uniformly over all target-window horizons programmed at time $t$.

Finally, letting $E_1$ denote the event in Assumption~A1 and $E_2$ denote the event $\{(\vartheta_k,\varphi_k)\in\mathcal{S}_k\}$ in Assumption~A2, we have
\begin{equation}
\mathbb{P}(E_1\cap E_2)
\ge
1-\mathbb{P}(E_1^c)-\mathbb{P}(E_2^c)
\ge
1-(\rho+\rho_s),
\end{equation}
which completes the proof.
\end{proof}

\section{Communication Performance Under Different User Numbers and Calibration Settings}

To complement the main manuscript evaluation, we provide additional communication performance under different user-number settings and calibration settings.
In addition to the default case $K=10$ reported in the main manuscript, we further test $K=8$ and $K=12$ to assess whether the proposed delay-aware forecast-then-beamform method remains effective when the users decreases or increases around the default operating point.
The simulation protocol follows the same setting as in the main manuscript, and only the user number or calibration setting is changed.

\subsection{User-Number Variation}

\begin{table}[t]
\centering
\footnotesize
\caption{Communication performance under different user-number settings.}
\label{tab:supp_usernum}
\begin{tabular}{lcc}
\toprule
Setting & QAR & Sum-rate \\
\midrule
$K=8$  & 0.2299 & 23.0020 \\
$K=10$ & 0.1994 & 25.8342 \\
$K=12$ & 0.1687 & 26.2766 \\
\bottomrule
\end{tabular}
\end{table}

Table~\ref{tab:supp_usernum} provides the communication performance results under different user-number settings compared to the main manuscript.
As the number of users increases, the admission difficulty naturally becomes higher, which reduces the QoS admission ratio (QAR).
Meanwhile, the sum-rate varies more mildly across the tested range.
This is because admitting more users does not necessarily improve feasibility, but the admitted set can still be efficiently refined under the same beamforming method.
Overall, these results provide evidence for the scalability of the algorithm proposed in the main manuscript.

\subsection{Calibration Sensitivity}

\begin{table}[t]
\centering
\footnotesize
\caption{Forecasting sensitivity to calibration confidence.}
\label{tab:supp_calib_pred}
\begin{tabular}{cccc}
\toprule
Calibration confidence $(1-\rho)$ & Interval coverage & MAE & RMSE \\
\midrule
0.80 & 0.9877 & 0.3722 & 0.6188 \\
0.85 & 0.9900 & 0.3722 & 0.6188 \\
0.90 & 0.9921 & 0.3722 & 0.6188 \\
0.95 & 0.9943 & 0.3722 & 0.6188 \\
0.99 & 0.9968 & 0.3722 & 0.6188 \\
\bottomrule
\end{tabular}
\end{table}

Table~\ref{tab:supp_calib_pred} reports the sensitivity of the offline calibration step to the confidence level $(1-\rho)$.
As $(1-\rho)$ increases, the interval coverage increases monotonically, confirming that a larger confidence level yields a more conservative calibrated bound.
By contrast, the macro MAE and RMSE remain unchanged across the sweep, since the calibration step only affects the post-hoc robustness margin and does not modify the point forecaster itself.
These results support the main-text discussion that calibration confidence mainly controls the conservativeness of the offline robustness certificate rather than the point-forecast accuracy.

\bibliography{IEEEabrv,Multimodal_Large_Language_Model_Enabled_Robust_Beamforming_for_HAP_Downlink_Communications}
\bibliographystyle{IEEEtran}
\end{document}